\documentclass[12pt,reqno]{amsart}
\usepackage[a4paper,top=3cm,bottom=3cm,inner=3cm,outer=3cm]{geometry}

\usepackage{amsmath, amsfonts, mathtools}

\usepackage{xcolor}
\definecolor{darkred}{rgb}{0.45,0,0}
\definecolor{lightblue}{rgb}{0.7,0.7,1}
\definecolor{darkorange}{rgb}{0.8, 0.4, 0.1}
\definecolor{darkgreen}{rgb}{0.05, 0.7, 0.06}
\definecolor{darkblue}{rgb}{0.2, 0.2, 0.9}
\definecolor{darkpink}{rgb}{1.0, 0.08, 0.58}
\definecolor{pink}{rgb}{1.0, 0.5, 0.8}

\usepackage[
    colorlinks, 
    citecolor=blue, 
    urlcolor=darkred, 
    final, 
    hyperindex, 
    pagebackref, 
    linkcolor = darkblue
]{hyperref}

\usepackage[capitalize]{cleveref}
\usepackage{slashed}

\usepackage{amsthm}

\theoremstyle{plain}
\newtheorem{thm}{Theorem}
\newtheorem*{thm*}{Theorem}

\newtheorem*{prop*}{Proposition}

\theoremstyle{definition}

\theoremstyle{remark}

\newcommand{\maps}{\colon}
\newcommand{\hf}{{\textstyle \frac{1}{2}}}
\newcommand{\fth}{{\textstyle \frac{1}{4}}}

\renewcommand{\d}{\slashed{\partial}} % Dirac
\newcommand{\V}{V} % spin-j representation of SU(2)
\renewcommand{\H}{\mathsf{H}}  % Hilbert space of spin-1/2 particle
\newcommand{\h}{\mathcal{H}} % arbitrary Hilbert space
\newcommand{\Fock}{\mathbf{\Lambda}}  % Fock space
\newcommand{\ad}{\mathrm{ad}}

\newcommand{\q}{\mathbf{q}}
\newcommand{\p}{\mathbf{p}}

\renewcommand{\L}{\mathbf{L}}
\newcommand{\M}{\mathbf{M}}
\newcommand{\A}{\mathbf{A}}
\newcommand{\B}{\mathbf{B}}
\newcommand{\e}{\mathbf{e}}

\newcommand{\Z}{\mathbb{Z}}
\newcommand{\R}{\mathbb{R}}
\newcommand{\C}{\mathbb{C}}
\newcommand{\HH}{\mathbb{H}}

\newcommand{\SU}{\operatorname{SU}}
\newcommand{\U}{\operatorname{U}}
\newcommand{\SO}{\operatorname{SO}}

\newcommand{\su}{\mathfrak{su}}
\newcommand{\so}{\mathfrak{so}}
\newcommand{\g}{\mathfrak{g}}
\renewcommand{\sl}{\mathfrak{sl}}

\usepackage{tikz}
\usetikzlibrary{shapes,calc}

\title{Second Quantization for the Kepler Problem}

\author[Baez]{John C.\ Baez} 

\address{School of Mathematics, University of Edinburgh, James Clerk Maxwell Building, Peter Guthrie Tait Road, Edinburgh, UK EH9 3FD}

\address{Department of Mathematics, University of California, Riverside CA, USA 92521}

\email{baez@math.ucr.edu}

\begin{document}

\begin{abstract}
The Kepler problem concerns a point particle in an attractive inverse square force.  After a brief review of the classical and quantum versions of this problem, focused on their hidden $\SU(2) \times \SU(2)$ symmetry, we discuss the quantum Kepler problem for a spin-$\hf$ particle.  We show that the Hilbert space $\mathcal{H}$ of bound states for this problem is unitarily equivalent, as a representation of $\SU(2) \times \SU(2)$, to the Hilbert space of solutions of the Weyl equation on the spacetime $\R \times S^3$.    This equation describes a massless chiral spin-$\hf$ particle.  We then form the fermionic Fock space on $\mathcal{H}$ and show this is unitarily equivalent to the Hilbert space of a massless chiral spin-$\hf$ free quantum field on $\R \times S^3$, again as representations of $\SU(2) \times \SU(2)$.  By modifying the Hamiltonian of this free field theory, we obtain the well-known `Madelung rules'.  These give a reasonable approximation to the observed filling of subshells as we consider elements with more and more electrons, and match the rough overall structure of the periodic table.
\end{abstract}

\maketitle

%\setcounter{tocdepth}{1} % comment this out if you want to see the subsections in the table of contents
%\tableofcontents

\section{Introduction}
\label{sec:introduction}

On January 6, 1680 Robert Hooke wrote to Isaac Newton, suggesting that he look into an
inverse square force law for gravity:
\begin{quote}
But my supposition is that the Attraction always is in a duplicate
proportion to the Distance from the Center Reciprocall [...]
\end{quote}
Ever since, the mathematics of the so-called Kepler problem---the motion of a particle in a central force proportional to the inverse square of distance---has revealed more and more interesting features.  Quantizing the Kepler problem was instrumental in testing quantum mechanics against the hydrogen atom, and it turns out this problem has a hidden symmetry under rotations in 4-dimensional space.   In fact, wavefunctions for bound states of the hydrogen can be described as functions on the 3-sphere, $S^3$.

Here we consider \emph{second quantization} for the Kepler problem.    In second quantization, a state in Fock space describes a collection of particles.  Thus, we can use it to describe multi-electron atoms.  This sets up a surprising relation between chemistry and quantum field theory: specifically, the massless spin-$\hf$ field on the spacetime $\R \times S^3$.   While the interactions between electrons are not accounted for in this approach, we can at least give a Hamiltonian for the massless spin-$\hf$ field that yields some well-known rules describing the structure of the periodic table of elements.

We start with a brief review.    The history of the inverse square law goes back at least to Kepler's discovery that the planets move in elliptical orbits, and in some sense even earlier to Apollonius of Perga's \textsl{Conics}.  But let us start with Newton.  We make no pretensions to completeness in this summary.

Newton discovered that the three kinds of orbits for a particle in an inverse square force---bound, unbound, and right on the brink between bound and unbound---are the three types of conic section: ellipse, hyperbola and parabola.   To modern eyes, the fact that the major axis of an elliptical orbit remains unchanged with the passage of time hints at extra conserved quantities besides the obvious ones.  Indeed, if we work in units where the inverse square force law says
\begin{equation}
    \ddot \q = - \frac{\q}{q^3}, 
\end{equation}
and if we define momentum by $\p = \dot \q$, then not only are energy
\begin{equation}
      E = \hf p^2 - \frac{1}{q}   
\end{equation}
and angular momentum
\begin{equation}
     \L = \q \times \p  
\end{equation}
conserved, but also the vector
\begin{equation}
 \e = \p \times \L - \frac{\q}{q}.  
 \end{equation}
This vector always points in the direction of the orbit's perihelion, and its magnitude $e$ equals the eccentricity of the orbit.  

This extra conserved quantity was named the `Runge--Lenz vector' after Lenz \cite{Lenz} used it in 1924 to study the hydrogen atom, citing Runge's \cite{Runge} popular textbook from five years earlier.  But Runge never claimed any originality: he attributed this vector to Gibbs.  Now many people call it the `Laplace--Runge--Lenz vector', honoring Laplace's \cite{Laplace} discussion of it in 1799.    But in fact this vector goes back at least to Jakob Hermann \cite{Hermann}, who wrote about it in 1710, triggering further work by Johann Bernoulli \cite{Bernoulli} in the same year.    Perhaps it would be wise to call $\e$ simply the `eccentricity vector'.

The power of this vector becomes apparent when we take its dot product with the particle's position $\q$.  A little manipulation gives
\begin{equation}
\label{eq:identity_1}
 \e \cdot \q = L^2 - q . 
\end{equation}
Combined with the fact that $\q$ moves in the plane perpendicular to $\L$, this equation describes a conic of eccentricity $e$.  We can also take the dot product of $\e$ with itself and show, using only vector identities and the formulas above, that
\begin{equation}
\label{eq:identity_2}
   e^2 = 1 + 2L^2 E. 
\end{equation}
This gives three cases:
\begin{itemize}
\item $E < 0$: in this case $e < 1$ and the orbit is an ellipse (perhaps a circle).
\item $E = 0$: in this case $e = 1$ and the orbit is a parabola.
\item $E > 0$: in this case $e > 1$ and the orbit is a hyperbola.
\end{itemize}
This paper focuses solely on the first case.

In 1847 Hamilton \cite{Hamilton} discovered another remarkable feature of the inverse square force law, whose full significance became clear only later: the momentum $\p$ moves in a circle.   This can seen using the conservation of $\L$ and $\e$.  Taking the inner product of 
\begin{equation}
 \frac{\q}{q} = \p \times \L - \e  
 \end{equation}
with itself, which is $1$, and doing some manipulations using the fact that $\L$ is perpendicular to both $\p$ and $\e$, we can show
\begin{equation}
\label{eq:Hamilton}
  \left( \p - \frac{\L \times \e}{L^2} \right)^2 = \frac{1}{L^2}.
\end{equation}
Thus, $\p$ stays on a circle centered at the point $(\L \times \e)/L^2$.

We now know that in classical mechanics, conserved quantities come from symmetries.   In the Kepler problem, conservation of energy comes from time translation symmetry, while conservation of the angular momentum comes from rotation symmetry.  Which extra symmetries give conservation of the eccentricity vector?  

A systematic approach to this question uses Poisson brackets.   If we use the sign convention where $\{q_j,p_k\} = \delta_{jk}$, some calculations give
\begin{equation}
\begin{array}{ccr}
      \{L_j, L_k\} &=& \epsilon_{jk\ell} L_\ell \\ [2pt]  \relax
      \{e_j, L_k\}  &=& \epsilon_{jk\ell} e_\ell \\ [2pt]  \relax
      \{e_j, e_k \}  &=& -2E \epsilon_{jk\ell} L_\ell 
\end{array}
\end{equation}
and of course
\begin{equation}
    \{E, L_j\} = \{E,e_j\} = 0 
\end{equation}
since $\L$ and $\e$ are conserved.   The factor of $-2E$ above is annoying, but on the region of phase space where $E < 0$---that is, the space of bound states, where the particle carries out an elliptical orbit---we can define a vector
\begin{equation}
  \M = \frac{\e}{\sqrt{-2E}} 
\end{equation}
and obtain
\begin{equation}
\begin{array}{ccc}    
\{L_j, L_k\} &=& \epsilon_{jk\ell} L_\ell  \\ [2pt]  \relax
\{L_j, M_k\} &=& \epsilon_{jk\ell} M_\ell  \\ [2pt]  \relax
\{M_j, M_k \} &=& \epsilon_{jk\ell} M_\ell .
\end{array}
\end{equation}
This gives a Lie algebra isomorphic to $\so(3) \oplus \so(3)$, as becomes clear if we
set
\begin{equation}
\label{eq:AB_from_LM}
\A = \hf(\L + \M), \qquad \B = \hf(\L - \M)  
\end{equation}
and check that
\begin{equation}
\label{eq:poisson_brackets}
\begin{array}{ccc}  
\{ A_j, A_k\} &=&  \epsilon_{jk\ell} A_\ell \\ [2pt]  \relax
\{ B_j, B_k\} &=&  \epsilon_{jk\ell} B_\ell  \\ [2pt] \relax
\{ A_j, B_k\} &=& 0 .
\end{array}
\end{equation}
But $\so(3) \oplus \so(3) \cong \so(4)$, so conservation of the eccentricity vector 
must come from a hidden $\so(4)$ symmetry.   And indeed, the group $\SO(4)$ acts on the bound states of the Kepler problem in a way that commutes with time evolution!

It seems that the first geometrical explanation of this symmetry was found in the quantum-mechanical context.  In 1926, even before Schr\"odinger came up with his famous equation, Pauli \cite{Pauli} used conservation of angular momentum and the eccentricity to determine the spectrum of hydrogen.   In 1935, Fock \cite{Fock} explained this symmetry by setting up an equivalence between hydrogen atom bound states and functions on the 3-sphere \cite{Fock}.  In the following year, Bargmann \cite{Bargmann} connected Pauli and Fock's work using group representation theory.  But it seems the first global discussion of this symmetry in the classical context was given by Bacry, Ruegg, and Souriau \cite{BacryRueggSouriau} in 1966, leading to important work by Souriau \cite{Souriau} and Moser \cite{Moser} in the early 1970s.    Since then, much more has been done.

The key to understanding the $\SO(4)$ symmetry for bound states of the Kepler problem turns out to be Hamilton's result about momentum moving in circles.  Hamilton's circles, defined by Equation \eqref{eq:Hamilton}, are not arbitrary circles in $\R^3$.  Using the inverse of stereographic projection, we can map $\R^3$ to the unit 3-sphere:
\begin{equation}
 \begin{array}{rccl}
  f \maps &\R^3 &\to & S^3    \subset \R^4  \\  \\
              & \p    &\mapsto &  
              \displaystyle{\left(\frac{p^2 - 1}{p^2 +1}, \frac{2 \p}{p^2 + 1}\right).}
\end{array}
\end{equation}
This map sends Hamilton's circles in $\R^3$ to great circles in $S^3$.  Furthermore, this construction gives all the great circles in $S^3$ except those that go through the points $(\pm 1, 0,0,0)$.  These missing great circles correspond to periodic orbits where a particle starts with momentum zero, falls straight to the origin, and bounces back the way it came.   We can embed the phase space of bound states of the Kepler problem, as a symplectic manifold, into $T^\ast S^3$ in such a way that these additional orbits become legitimate trajectories in this larger phase space.   Neglecting these additional orbits caused trouble for Bohr and Sommerfeld in their early work on quantizing the hydrogen atom \cite{Bucher}.   For a clear modern treatment of Hamilton's circles, see Milnor \cite{Milnor} and Egan \cite{Egan}.

Note that points of $S^3$ correspond not to positions but to \emph{momenta} in the Kepler problem.    As time passes, these points move along great circles in $S^3$.  How is their dynamics related to geodesic motion on the 3-sphere?  We can understand this as follows.   From Equation \eqref{eq:identity_2} and the definition of $\M$ it follows that
\begin{equation}
  L^2 + M^2 =  - \frac{1}{2E}, 
\end{equation}
and using the fact that $\L \cdot \M = 0$, an easy calculation gives
\begin{equation}
\label{eq:hamiltonian_classical}
  E \; = \; -\frac{1}{8A^2} \; = \; -\frac{1}{8B^2}.
\end{equation}
In the 3-sphere picture, the observables $A_j$ become functions on $T^\ast S^3$.  These functions are just the components of momentum for a particle on $S^3$, defined using a standard basis of right-invariant vector fields on $S^3 \cong \SU(2)$.   Similarly, the observables $B_j$ are the components of momentum using a standard basis of left-invariant vector fields.     It follows that
\begin{equation}
\label{eq:hamiltonian_reciprocal}
K = 8A^2 = 8B^2 
\end{equation}
is the Hamiltonian for a nonrelativistic free particle on $S^3$ with an appropriately chosen mass.    Such a particle moves around a great circle on $S^3$ at constant speed.   Since the Kepler Hamiltonian $E$ is the negative reciprocal of $K$, particles governed by \emph{this} Hamiltonian move along the same trajectories---but typically not at constant speed. 

Both $K$ and the Kepler Hamiltonian $E = -1/K$ are well-defined smooth functions on the symplectic manifold that Souriau \cite{Souriau} dubbed the `Kepler manifold':
\begin{equation}
   T^+ S^3 = \{ (x,p) : \; x \in S^3, p \in T_x S^3, p \ne 0 \}  .
\end{equation}
We can also think of $T^+ S^3$ as a space of light rays in the `Einstein universe': the manifold $\R \times S^3$ with Lorentzian metric $dt^2 - ds^2$, where $ds^2$ is the usual metric on the unit sphere.   Here a `light ray' is a null geodesic equipped with a choice of covariantly constant tangent vector field, its `4-velocity'.   The corresponding cotangent vector field describes the energy-momentum of the light ray.  This extra information reflects the fact that massless particles can have different energy and momentum even if they trace out the same path in spacetime.     To describe a light ray using a point in $T^+ S^3$, we let $x \in S^3$ be the light ray's position at time zero, while the null cotangent vector $p + \|p\| dt$ describes the light ray's energy-momentum at time zero.  In this way $T^+ S^3$ serves as a phase space for a classical free massless spin-$0$ particle in the Einstein universe.  The Hamiltonian for such a particle is $\sqrt{K}$.   

In what follows we start with the quantum Kepler problem described in terms of the 3-sphere, and then carry out \emph{second quantization} to treat multi-electron atoms, also bringing the electron's \emph{spin} into the picture.   These two novel features are linked: only by treating electrons as identical spin-$\hf$ particles can we obtain the usual picture where, thanks to the Pauli exclusion principle, at most two electrons occupy each orbital. 

We begin in Section \ref{sec:spinless} by reviewing Fock's $S^3$ description of bound states of the hydrogen atom.  This treatment ignores the electron's spin, which we introduce in Section \ref{sec:spin}.    In Section \ref{sec:geometry} we express the Hamiltonian for bound states of the hydrogen atom with spin-$\hf$ electron in terms of  the Dirac operator on $S^3$.  We also give a quaternionic description of these bound states.   In Section \ref{sec:einstein} we reinterpret these bound states as states of a massless chiral spin-$\hf$ particle in the Einstein universe.  In Section \ref{sec:second} we apply second quantization, showing how states of multi-electron atoms correspond to states of a massless chiral spin-$\hf$  free quantum field on the Einstein universe.   Finally, in Section \ref{sec:aufbau} we describe a modified Hamiltonian for this quantum field that makes the lowest-energy $N$-particle state approximately follow the rules summarized by the periodic table---the so-called `Madelung rules'.   

\subsection{Note on units}
\label{subsec:units}

Throughout this paper when discussing the hydrogen atom we work in units with
\[   \hbar = e = 4 \pi \epsilon_0 = \mu = 1 \]
where
\begin{itemize}
\item $\hbar$ is Planck's constant,
\item $-e$ is the charge of the electron,
\item $\epsilon_0$ is the permittivity of the vacuum,
\item $\mu = m_e/(m_e + m_p)$ is the reduced mass of the electron.
\end{itemize}
This allows us to hide all these fundamental constants and write the classical Hamiltonian for the hydrogen atom as in Equation \eqref{eq:hamiltonian_classical}.  For a nucleus with atomic number $N$ we instead set $N e = 1$.

\subsection{Acknowledgements}

I thank my Ph.D.\ advisor, Irving Segal, for teaching me the ways of mathematical physics and introducing me to the conformal geometry of the Einstein universe.  I also thank Greg Egan and Paul Schwahn for help with this project, J.\ Gregory Moxness for Figure \ref{figure_1}, and the referee for useful comments.

\section{The hydrogen atom --- ignoring spin}
\label{sec:spinless}

As already mentioned, there are various ways to quantize the Kepler problem and obtain a description of the hydrogen atom's bound states as wavefunctions on the 3-sphere.   Here we take a less systematic but quicker approach.  We start with the space of wavefunctions on the 3-sphere, $L^2(S^3)$, and describe operators on this space arising from rotational symmetries.   Combining these with standard facts about the hydrogen atom---which we take as known, rather than derive---we set up a unitary equivalence between $L^2(S^3)$ and the space of hydrogen atom bound states, and show the hydrogen atom Hamiltonian has $\SO(4)$ symmetry.  Everything in this section is essentially a review of known material.   Here we ignore the electron's spin, which we introduce in the next section. 

We identify the 3-sphere with the Lie group $\SU(2)$.  The group $\SU(2)$ acts on itself in three important ways, which combine to give the mathematics we need:
\begin{itemize}
\item left multiplication by $g$: $g$ maps $h$ to $gh$,
\item right multiplication by $g^{-1}$: $g$ maps $h$ to $hg^{-1}$,
\item conjugation by $g$: $g$ maps $h$ to $ghg^{-1}$
\end{itemize}
Left and right multiplication commute, so they combine to give a left action of $\SU(2) \times \SU(2)$ on $\SU(2)$.     We thus obtain a unitary representation $R$ of $\SU(2) \times \SU(2)$ on $L^2(\SU(2))$, given by
\begin{equation}
    (R(g_1, g_2) \psi)(g) = \psi(g_1^{-1} g g_2). 
\end{equation}
In what follows we denote $\SU(2)$ as $S^3$ when we regard it as the unit sphere in $\R^4$, 
acted on by $\SU(2) \times \SU(2)$ as above.  The Peter--Weyl theorem lets us decompose $L^2(S^3)$ into finite-dimensional irreducible unitary representations of $\SU(2) \times \SU(2)$:
\begin{equation}
\label{eq:decomposition_1}
    L^2(S^3) \cong \bigoplus_{j} \V_j \otimes \V_j .
\end{equation}
Here $j = 0, \frac{1}{2}, 1, \frac{3}{2}, \dots$ and $\V_j$ is the spin-$j$ representation of $\SU(2)$, which is the irreducible unitary representation of dimension $2j + 1$.  

The representation of $\SU(2) \times \SU(2)$ on $L^2(S^3)$ is generated by self-adjoint operators that correspond to familiar observables for bound states of the hydrogen atom.  
The complexification of $\su(2) \oplus \su(2)$ has self-adjoint elements
\begin{equation}
A_j = (\hf \sigma_j, 0), \qquad 
B_j =  (0, \hf \sigma_j) 
\end{equation}
where $\sigma_j$ are the Pauli matrices.   We also use $A_j$ and $B_j$ to denote the
corresponding self-adjoint operators on $L^2(S^3)$.  They obey the following commutation relations, which are the quantum-mechanical analogues of the Poisson brackets in Equation \eqref{eq:poisson_brackets}:
\begin{equation}
\label{eq:AB_commutation}
\begin{array}{ccc}  
[A_j, A_k] &=&  i\epsilon_{jk\ell} A_\ell \\ [2pt]  \relax
[B_j, B_k] &=&  i\epsilon_{jk\ell} B_\ell  \\ [2pt] \relax
[A_j, B_k] &=& 0 .
\end{array}
\end{equation}
Geometrically speaking, the skew-adjoint operators $-iA_j$ act on $L^2(S^3)$ as differentiation by vector fields on $S^3 \cong \SU(2)$ that generate left translations, while the operators $-iB_j$ act by differentiation by vector fields that generate right translations.    We also use $-iA_j$ and $-iB_j$ to stand for these vector fields.   Since left-invariant-vector fields generate right translations and vice versa, the vector fields $-iA_j$ are right-invariant, while the $-iB_j$ are left-invariant.

As well known, we have
\begin{equation}
\label{eq:casimir}
  \phi \in \V_j \otimes \V_j \implies  A^2 \phi = B^2 \phi = j(j+1) \phi .
\end{equation}
This implies that $A^2 = B^2$ on all of $L^2(S^3)$.   We can define an operator on $L^2(S^3)$ that corresponds to the Hamiltonian for bound states of the hydrogen atom, namely 
\begin{equation}
\label{eq:hamiltonian_without_spin}
H_0 \;=\; - \frac{1}{8(A^2 + \fth)} \;=\;- \frac{1}{8(B^2 + \fth)} .
\end{equation}
This is the quantum analogue of the classical Hamiltonian in Equation \eqref{eq:hamiltonian_classical}.    We discuss the curious appearance of the number $\fth$ here in Subsection \ref{subsec:duflo}, but we can easily see why it is needed.   On the subspace $\V_j \otimes \V_j$, the operator $4A^2 + 1$ acts as multiplication by
\[    4j(j+1) + 1 = 4j^2 + 4j + 1 = (2j + 1)^2. \]
In atomic physics it is traditional to work with the dimension of $\V_j$, 
\begin{equation}
 n = 2j + 1,
\end{equation}
rather than $j$ itself.   Thus, we have 
\begin{equation}
  \phi \in \V_j \otimes \V_j \implies   H_0\phi = - \frac{1}{2n^2} \phi ,
\end{equation}
and as $j$ ranges over all allowed values $j = 0, \frac{1}{2}, 1, \dots$, $n$ ranges over all positive integers.  These are precisely the usual energy eigenvalues for the hydrogen atom, expressed in the units chosen in Subsection \ref{subsec:units}.

We have not yet brought in the action of $\SU(2)$ on $S^3$ by conjugation.   This gives yet another unitary representation of $\SU(2)$ on $L^2(S^3)$.   To conjugate by $g$ we both left multiply by $g$ and right multiply by $g^{-1}$.  Thus, the conjugation representation of $\SU(2)$ on $L^2(S^3)$ has self-adjoint generators
\begin{equation}
     L_j = A_j + B_j  
\end{equation}
and these obey
\begin{equation}
      [L_j, L_k] = i\epsilon_{jk\ell} L_\ell  .
\end{equation}
These operators $L_j$ are the components of the angular momentum of the hydrogen, as we can see from the classical picture in Equation \eqref{eq:AB_from_LM}.   With respect to the conjugation representation, the subspace $\V_j \otimes \V_j \subset L^2(S^3)$ decomposes according to the Clebsch--Gordan rules:
\begin{equation}
\label{eq:decomposition_2}
  \V_j \otimes \V_j \cong \bigoplus_{\ell = 0}^{2j} \V_\ell  
\end{equation}
where $\ell$ takes integer values going from $0$ to $2j = n - 1$.   Thus, with respect to the conjugation representation, we can decompose $L^2(S^3)$ into irreducible representations of $\SU(2)$ by combining Equations \eqref{eq:decomposition_1} and \eqref{eq:decomposition_2}, obtaining
\begin{equation}
\label{eq:decomposition_3}
   L^2(S^3) \cong \bigoplus_{n = 1}^\infty \bigoplus_{\ell = 0}^{n-1} \V_\ell.
\end{equation}

The summand $\V_\ell$ has a basis of eigenvectors for $L_3$ with eigenvalues taking integer values $m$ ranging from $-\ell$ to $\ell$.      Thus $L^2(S^3)$ has an orthonormal basis of states $|n , \ell , m\rangle$ where:
\begin{itemize}
\item $n$ ranges over positive integers;
\item $\ell$ ranges from $0$ to $n-1$ in integer steps;
\item $m$ ranges from $-\ell$ to $\ell$ in integer steps.  
\end{itemize}
From the calculations thus far we have
\begin{equation}
\label{eq:hydrogen_spinless}
\begin{array}{ccl}  
A^2 |n, \ell, m \rangle &=& B^2 |n, \ell, m \rangle \; = \;
 \fth(n^2-1) |n, \ell, m \rangle  \\  [8pt]
H_0 |n, \ell, m \rangle &=& \displaystyle{ - \frac{1}{2n^2} \, |n, \ell, m \rangle } \\ [8pt]
    L^2 |n, \ell, m \rangle &=& \ell(\ell + 1) \, |n , \ell, m \rangle  \\ [3pt]
    L_3 |n , \ell, m \rangle &=& m \, |n , \ell, m \rangle  .
\end{array}
\end{equation}
The last three relations are familiar from work on the hydrogen atom.  In this context
\begin{itemize}
\item  $n$ is the `principal quantum number',
\item $\ell$ is the `azimuthal quantum number', 
\item $m$ is the `magnetic quantum number'. 
\end{itemize} 
The operator $H_0$ corresponds to the Hamiltonian of the hydrogen atom, with the electron treated as spinless, and the states $|n, \ell, m\rangle$ are a well-known basis of the hydrogen atom's bound states.  

\subsection{Digression on the Duflo isomorphism}
\label{subsec:duflo}

Why is there an extra $\fth$ in the quantum Hamiltonian
\[
H \;=\; - \frac{1}{8(A^2 + \fth)} \;=\;- \frac{1}{8(B^2 + \fth)} 
\]
not present in the classical case?  From a pragmatic viewpoint, we need this to match the spectrum obtained in the usual approach to quantizing the Kepler problem, which agrees with experiment.  We also need some term like this to avoid dividing by zero.    But there is yet another explanation: it arises from the Duflo isomorphism.

Suppose $G$ is a connected Lie group with Lie algebra $\g$.  The Poincar\'e--Birkhoff--Witt theorem gives a natural linear map from the polynomial algebra $S(\g)$ to the universal enveloping algebra $U(\g)$.  This map is an isomorphism of vector spaces, but clearly not of algebras: $S(\g)$ is commutative while $U(\g)$ is not.   This map is compatible with the natural representation of $G$ on these spaces, so it restricts to a vector space isomorphism 
\begin{equation}
    F \maps S(\g)^G \to U(\g)^G 
\end{equation}
where the superscript indicates the $G$-invariant subspace.   Both  $S(\g)^G$ and $U(\g)^G$ are commutative algebras, and indeed $U(\g)^G$ is the center of $U(\g)$.  However, $F$ is not an algebra homomorphism.   Nonetheless, Duflo \cite{Duflo} proved that in some cases we can compose $F$ with a linear map 
\begin{equation}
 E \maps S(\g)^G \to S(\g)^G 
\end{equation}
to get an algebra isomorphism
\[   F \circ E \maps S(\g)^G \to U(\g)^G .\]
Later Kontsevich  \cite{Kontsevich} showed that this works for all finite-dimensional Lie algebras.  Calaque and Rossi have written a useful pedagogical account \cite{CalaqueRossi}.   

The map $E$ is defined as follows.   First, for any $x \in \g$ we have a linear map 
\begin{equation}
\begin{array}{cccl}
\ad_x \maps & \g &\to& \g  \\
                       & y & \mapsto & [x,y] 
\end{array}
\end{equation}
Second, note that linear functionals on $\g$ are vectors in $\g^\ast$, so they give constant vector fields on $\g^\ast$.   More generally, polynomial functions on $\g$ give constant-coefficient differential operators on $\g^\ast$.   These differential operators act on $S(\g)$, since this can be identified with the algebra of polynomial functions on $\g^\ast$.   Even more generally, entire functions on $\g$ such as
\begin{equation}
  \tilde{J}(x) =  \det\left( \frac{e^{\ad_x/2} - e^{-\ad_x/2}}{\ad_x}\right)  
\end{equation}
act as operators on $S(\g)$, since we can expand them as Taylor series, and only finitely many terms give a nonzero result when applied to any particular element of $S(\g)$.   Since $\tilde{J}$ is $G$-invariant, the operator on $S(\g)$ coming from this particular function restricts to an operator on $S(\g)^G$, which is the desired linear map $E \maps S(\g)^G \to S(\g)^G$.

In the case $\g = \sl(2,\C)$, the algebra $S(\g)^G$ is generated by the element $J_1^2 + J_2^2 + J_3^2$ where $J_j = \tfrac{i}{2} \sigma_j$.   Applying $E$ to this element we obtain
\begin{equation}
    \tilde{J}^2 = J_1^2 + J_2^2 + J_3^2 + \fth .
\end{equation}
This explains the extra $\fth$ in the hydrogen atom Hamiltonian  \cite{RosaVitale}.

Examining the details of the map $E$, one can see that the value $\fth$ arises from the fact that
\begin{equation}
\label{24}
 \frac{e^{x/2} - e^{-x/2}}{x} = 1 + \frac{x^2}{24} + \cdots 
\end{equation}
together with the fact that each of the three elements $J_j^2$ has trace $2$ in the adjoint representation, giving a correction of
\[      3 \times 2 \times \tfrac{1}{24} = \fth .\]
In fact, the appearance of the number $24$ in Equation \eqref{24} is closely connected to role of that number in conformal field theory, the theory of modular forms, topology, and other subjects.    Thus, the $\fth$ in the hydrogen atom Hamiltonian is part of a much larger story.  

\section{The hydrogen atom --- with spin}
\label{sec:spin}

Next we include the electron's spin in our treatment of the hydrogen atom.  To do this, we merely tensor the Hilbert space of the previous section, $L^2(S^3)$, with a copy of $\C^2$ describing the electron's spin.  The resulting space $L^2(S^3) \otimes \C^2$ is the Hilbert space of bound states of a \emph{spinor-valued} version of the Schr\"odinger equation for the hydrogen atom.  This is simplification of a more careful treatment using the Dirac equation: it neglects all spin-dependent terms in Hamiltonian, such as spin-orbit interactions.   These spin-dependent terms give corrections that go to zero in the limit where the speed of light approaches infinity.  In this sense, we are giving a nonrelativistic treatment of the hydrogen atom, but taking into account the fact that the electron is a spin-$\hf$ particle.

The Hilbert space $L^2(S^3) \otimes \C^2$ becomes a unitary representation of $\SU(2)$ in three important ways.  The first two come from the actions of $\SU(2)$ on $L^2(S^3)$ by left and right translation, as described in the previous section.  The third comes from the natural action of $\SU(2)$ on $\C^2$.   All three of these actions of $\SU(2)$ on $L^2(S^3) \otimes \C^2$ commute with each other.  We thus get a unitary representation of $\SU(2) \times \SU(2) \times \SU(2)$ on $L^2(S^3) \otimes \C^2$.

It is useful to spell this out at the Lie algebra level.  In the previous section we introduced self-adjoint operators $A_j$ and $B_j$ on $L^2(S^3)$ with commutation relations given in Equation \eqref{eq:AB_commutation}.  These are the self-adjoint generators of the left and right translation actions of $\SU(2)$, respectively.   We now tensor these operators with the identity on $\C^2$ and obtain operators on $L^2(S^3) \otimes \C^2$, which by abuse of notation we denote with the same names: $A_j$ and $B_j$.  We also introduce `spin angular momentum' operators
\begin{equation}
  S_j = 1 \otimes \hf \sigma_j   
\end{equation}
on $L^2(S^3) \otimes \C^2$.   These obey the following commutation relations:
\begin{equation}
\label{eq:commutation}
\begin{array}{cclcccl}  
[A_j, A_k] &=&  i\epsilon_{jk\ell} A_\ell  &\quad &  [A_j, B_k] &=& 0 \\ [2pt]   \relax
[B_j, B_k] &=&  i\epsilon_{jk\ell} B_\ell &&  [A_j, S_k] &=& 0  \\ [2pt] \relax
[S_j, S_k] &=&  i\epsilon_{jk\ell} S_\ell && [B_j, S_k] &=& 0.
\end{array}
\end{equation}
We define `orbital angular momentum' operators
\begin{equation}
     L_j = A_j + B_j  
 \end{equation}
which are just those of the previous section tensored with the identity on $\C^2$.
These obey 
\begin{equation}
\begin{array}{ccl}  
[L_j, L_k] &=&  i\epsilon_{jk\ell} L_\ell \\  [2pt] \relax
 [S_j, S_k] &=& i \epsilon_{jk\ell} S_\ell \\  [2pt] \relax
[L_j, S_k] &=&  0  .
\end{array}
\end{equation}
We also define `total angular momentum' operators
\begin{equation}
   J_j = L_j + S_j   
\end{equation}
which obey
\begin{equation}
    [J_j, J_k] = i \epsilon_{jk\ell} J_\ell .
\end{equation}
Finally, we define a Hamiltonian for the hydrogen atom 
\begin{equation}
\label{eq:hamiltonian_quantum_1}
   H \; = \; - \frac{1}{8(A^2 + \fth)} \; = \; - \frac{1}{8(B^2 + \fth)}  
\end{equation}
which commutes with all the operators $A_j, B_j, S_j$ and thus also $L_j$ and $J_j$.   We have 
\begin{equation}
  H = H_0 \otimes 1  
\end{equation}
where $H_0$ is the Hamiltonian for the hydrogen atom with a spin-0 electron, as in Equation \eqref{eq:hamiltonian_without_spin}.   The spectrum of $H$ is the same as that of $H_0$; only the multiplicity of each eigenvalue has doubled.

From the direct sum decomposition in Equation \eqref{eq:decomposition_3} we obtain
\begin{equation}
\label{eq:decomposition_4}
   L^2(S^3) \otimes \C^2 \cong 
   \bigoplus_{n = 1}^\infty \bigoplus_{\ell = 0}^{n-1} \V_\ell \otimes \C^2 . 
\end{equation}
The basis $|n, \ell, m \rangle$ of $L^2(S^3)$ tensored with the standard basis of $\C^2$ gives an orthonormal basis $|n , \ell, m, s \rangle$ of $L^2(S^3) \otimes \C^2$ where:
\begin{itemize}
\item the principal quantum number $n$ ranges over positive integers;
\item the azimuthal quantum number $\ell$ ranges from $0$ to $n-1$ in integer steps;
\item the magnetic quantum number $m$ ranges from $-\ell$ to $\ell$ in integer steps; 
\item the spin quantum number $s$ is $+\frac{1}{2}$ or $-\frac{1}{2}$.  
\end{itemize}
Extending Equation \eqref{eq:hydrogen_spinless} to this situation we see
\begin{equation}
\label{eq:hydrogen}
 \begin{array}{ccl}  
 A^2 |n, \ell, m, s \rangle &=& B^2 |n, \ell, m, s \rangle \; =  \; 
 \fth( n^2 - 1) |n, \ell, m, s\rangle  \\  [8pt]
H |n, \ell, m, s \rangle &=& \displaystyle{ - \frac{1}{2n^2}\,  |n, \ell, m, s \rangle } \\ [12pt]
    L^2 |n, \ell, m, s\rangle &=& \ell(\ell + 1) |n , \ell, m, s \rangle  \\ [3pt]
    L_3 |n , \ell, m, s \rangle &=& m |n , \ell, m, s \rangle  \\ [3pt]
    S^2 |n , \ell, m, s \rangle &=& \frac{3}{4} |n , \ell, m, s \rangle  \\  [3pt]
    S_3 |n , \ell, m, s \rangle &=& s |n , \ell, m, s \rangle .
\end{array}
\end{equation}

Combining this with the textbook treatment of the hydrogen atom, it follows that $L^2(S^3) \otimes \C^2$ is indeed unitarily equivalent to the subspace of $L^2(\R^3) \otimes \C^2$ consisting of bound states of the spinor-valued Schr\"odinger equation
\begin{equation}
   i \frac{\partial \psi}{\partial t} = -\frac{1}{2} \nabla^2 \psi - \frac{1}{r} \psi  
\end{equation}
with the operators $H, L_j$ and $S_j$ having these definitions:
\begin{equation}
\begin{array}{ccl}    
H &=& \displaystyle{ -\frac{1}{2} \nabla^2 - \frac{1}{r} }  \\ [12pt]
L_j &=&  \displaystyle{ -i\epsilon_{jk\ell} x_k \frac{\partial}{\partial x_\ell} } \\ [10pt]
S_j &=& \frac{1}{2} \sigma_j   .
\end{array}
\end{equation}
In short, the operator defined in Equation \eqref{eq:hamiltonian_quantum_1} is
unitarily equivalent to the Hamiltonian on bound states of the hydrogen atom defined above.

\section{Spinors on the 3-sphere}
\label{sec:geometry}

In preparation for relating the hydrogen atom to quantum field theory on $\R \times S^3$, we now bring out the geometrical content of the previous section.  In Section \ref{sec:spinless} we studied the operator $A^2 = B^2$ on $L^2(S^3)$.    Now we shall see that this operator is proportional to the Laplacian on the unit 3-sphere.   We can think of elements of $L^2(S^3) \otimes \C^2$ as spinor fields on $S^3$ if we trivialize the spinor bundle using the  action of $\SU(2)$ as right translations on $S^3 \cong \SU(2)$.   Kronheimer \cite{Kronheimer} discussed the Dirac operator $\d$ on these spinor fields.   We recall this here and show that the hydrogen atom Hamiltonian, thought of as an operator on $L^2(S^3) \otimes \C^2$, is 
\[    H = - \frac{1}{2 (\d - \hf)^2}   . \]
We also determine the eigenvalues of the Dirac operator on the 3-sphere.

We begin with the Laplacian on the 3-sphere.     Recall from Section \ref{sec:spinless} that $-iB_j$ is a basis of left-invariant vector fields on $S^3$.    Each vector field $-iB_j$ gives a tangent vector at the identity of $\SU(2)$, namely $-\frac{i}{2} \sigma_j \in \su(2)$.    What is the length of this vector if we give $\SU(2)$ the usual Riemannian metric on the unit 3-sphere?   Exponentiating this vector we get $\exp(-\frac{i}{2} \sigma_j t)$, which is the identity precisely when $t$ is an integer times $4\pi$.  Since a great circle on the unit sphere has circumference $2\pi$, this vector must have length $\frac{1}{2}$.     It follows that the vector fields
\begin{equation}
\label{eq:X_vs_A}
       X_j = -2i B_j 
\end{equation}
have unit length everywhere, and one can check that they form an orthonormal basis of vector fields on $S^3$.   We thus define the (positive definite) Laplacian on $S^3$ to be the differential operator
\begin{equation}
\label{eq:laplacian}        
 \Delta =  - \sum_{i = 1}^3 X_j^2  = 4B^2.
\end{equation}

Combining Equations \eqref{eq:casimir} and \eqref{eq:laplacian}, we see that $\Delta$ acts as multiplication by $4j(j+1)$ on the subspace  $\V_j \otimes \V_j \subset L^2(S^3)$.   Setting $n = 2j+1$ as usual,
\[  4j(j+1) =  4j^2 + 4j  = n^2 - 1\]
so 
\begin{equation}
\label{eq:laplacian_eigenvalues}
     \phi \in \V_j \otimes \V_j \implies \Delta \phi = (n^2 - 1) \phi
\end{equation}
and the eigenvalues of the Laplacian on $L^2(S^3)$ are $n^2 - 1$ where $n$ ranges over all positive integers.  

Tensoring $\Delta$ with the identity we obtain a differential operator on $L^2(S^3) \otimes \C^2$, which by abuse of notation we again call $\Delta$.   This has the same spectrum as the Laplacian on $L^2(S^3)$.   Equations \eqref{eq:hamiltonian_quantum_1} and \eqref{eq:laplacian} then give this formula for the Hamiltonian of the hydrogen atom as an operator on $L^2(S^3) \otimes \C^2$:
\begin{equation}
\label{eq:hamiltonian_quantum_2}
    H =  - \frac{1}{2(\Delta + 1)} .
\end{equation}

Next we turn to the Dirac operator \cite{LawsonMichelson}.   Up to isomorphism there is only one choice of spin structure on $S^3$.   We can trivialize the tangent bundle of $S^3 \cong \SU(2)$ using left translations.   This lets us identify the oriented orthonormal frame bundle of $S^3$ with the trivial bundle $S^3 \times \SO(3) \to S^3$.  This gives a way to identify the spin bundle on $S^3$ with the trivial bundle $S^3 \times \SU(2) \to S^3$.  This in turn lets us identify spinor fields on $S^3$ with $\C^2$-valued functions.     

There are at least two important connections on the tangent bundle of $S^3$.  One is the Cartan connection: a vector field is covariantly constant with respect to this connection if and only if it is invariant under left translations on $S^3 \cong \SU(2)$.    The other is the Levi--Civita connection, which is the unique torsion-free connection for which parallel translation preserves the metric.  Parallel translation with respect to the Cartan connection also preserves the metric, but the Cartan connection is flat and has torsion, while the Levi--Civita connection is curved and torsion-free.   

Each of these connections lifts uniquely to a connection on the spin bundle and then gives a Dirac-like operator.  The Cartan connection gives covariant derivative operators $\nabla^c_j$ on $L^2(S^3) \otimes \C^2$ with 
\begin{equation}
  \nabla^c_j = X_j \otimes 1  
\end{equation}
while the Levi--Civita connection gives covariant derivative operators $\nabla_j$ with
\begin{equation}
   \nabla_j = X_j \otimes 1 + 1 \otimes \tfrac{i}{2}\sigma_j .
\end{equation}
We can define a self-adjoint version of the Dirac operator $\d$ on $L^2(S^3) \otimes \C^2$ using the Levi--Civita connection:
\begin{equation}
       \d = (1 \otimes (-i\sigma_j)) \nabla_j ,
\end{equation}
where as usual we sum over repeated indices.   On the other hand, Kronheimer \cite{Kronheimer} defined a Dirac-like operator $D$ using the Cartan connection:
\begin{equation}
\label{eq:Dirac-like_operator}
          D = (1 \otimes  (-i\sigma_j)) \nabla^c_j.
\end{equation}
An easy calculation shows how $\d$ and $D$ are related:
\begin{equation}
\label{eq:dirac_operators}
\begin{array}{ccl}
\d &=&   (1 \otimes (-i\sigma_j)) \nabla_j  \\ [3pt]
&=& (1 \otimes (-i\sigma_j)) \left( \nabla^c_j  + 1 \otimes \tfrac{i}{2}\sigma_j\right) \\ [3pt]
&=& D + \frac{3}{2}  
\end{array}
\end{equation}
where we use $\sigma_j^2 = 1$ and the 3-dimensionality of space.  

Let us compute $D^2$.   Using the identities
\begin{equation}
  \sigma_j \sigma_k = \delta_{jk} + i \epsilon_{jk\ell} \sigma_\ell , \quad 
  \epsilon_{jk\ell} X_j X_k = 2X_\ell   , \quad D = X_j \otimes (-i \sigma_j),
\end{equation}
we obtain
\begin{equation} 
\begin{array}{ccl}
D^2 &=& -X_j X_k \otimes \sigma_j \sigma_k   \\   [2pt]
&=& -X_j X_k \otimes (\delta_{jk} + i \epsilon_{jk\ell} \sigma_\ell )  \\   [2pt]  
&=& -X_j X_j \otimes 1 \, - \, 2i X_\ell \otimes \sigma_\ell  \\   [2pt]
&=& \Delta - 2 D.
\end{array}
\end{equation}
It follows that $\Delta =  D(D+2)$, so
\begin{equation}
\label{eq:Delta_D_and_dirac}
 \Delta + 1 = (D+1)^2 = (\d - \hf)^2 .
\end{equation}
Combining this fact with Equation \eqref{eq:hamiltonian_quantum_2} we can express the hydrogen atom Hamiltonian in terms of the Dirac operator on the 3-sphere:
\begin{equation}
\label{eq:hamiltonian_quantum_3}
H  \; = \; - \frac{1}{2(\d - \hf)^2} .
\end{equation}

Next let us study the eigenvectors and eigenvalues of the Dirac operator.    From Equation \eqref{eq:decomposition_1} we have
\begin{equation}
\label{eq:decomposition_4}
   L^2(S^3) \otimes \C^2 \cong \bigoplus_j  \V_j \otimes \V_j \otimes \C^2   
\end{equation}
where $j = 0, \hf, 1, \frac{3}{2}, \dots$ and $V_j$ is the spin-$j$ representation of $\SU(2)$.
Since $\d$ maps each finite-dimensional subspace $\V_j \otimes \V_j \otimes \C^2$ to itself and is self-adjoint on these subspaces, each of these subspaces has an orthonormal basis of eigenvectors.  Suppose $\psi \in \V_j \otimes \V_j \otimes \C^2$ has
\[    \d \psi = \lambda \psi. \]
Then by Equations \eqref{eq:laplacian_eigenvalues} and \eqref{eq:Delta_D_and_dirac} we have
\[   (\lambda - \hf)^2 \psi = (\d - \hf)^2 \psi = (\Delta + 1)\psi = n^2 \psi \]
where $n = 2j+1$, so $\lambda - \hf = \pm n $.  Thus, the only eigenvalues of $\d$ on the subspace $\V_j \otimes \V_j \otimes \C^2$ are $\pm n + \hf$, or in other words
\begin{equation}
\label{eq:dirac_eigenvalues}
    \psi \in \V_j \otimes \V_j  \otimes \C^2 \implies \d \psi = 
    \lambda \psi \text{\; for \; } \lambda = \pm (2j+1) + \hf  .
\end{equation}

We could go further and explicitly diagonalize the Dirac operator on $S^3$; for details see Di Cosmo and Zampini \cite{DiCosmoZampini}.  Instead we summarize two results from Kronheimer \cite{Kronheimer} which do more to clarify the overall picture.

First, Kronheimer shows that the spectrum of $\d$ is symmetric about the origin. 
To do this he identifies $\C^2$ with the quaternions, and thus $L^2(S^3) \otimes \C^2$ with 
a space of quaternion-valued functions on the 3-sphere; then quaternionic conjugation 
gives a conjugate-linear operator
\begin{equation}
  \dagger \maps L^2(S^3) \otimes \C^2 \to L^2(S^3) \otimes \C^2 
\end{equation}
with $\dagger^2 = 1$.  He then proves a result about the operator $D$ that implies
\begin{equation}  
 \d \psi = \lambda \psi \; \iff  \; \d \psi^\dagger = -\lambda \psi^\dagger .
\end{equation}

Second, he proves a result about the operator $D$ that implies that the eigenspace
\begin{equation}
\label{eq:dirac_eigenspaces}
    F_\lambda = \{ \psi \in L^2(S^3) \otimes \C^2 : \; \d \psi = \lambda \psi \} 
\end{equation}
has dimension
\begin{equation}
\label{eq:dimension_of_dirac_eigenspaces}
    \dim F_\lambda = (\lambda+\hf)(\lambda-\hf) 
\end{equation}
when $\lambda \in \Z + \hf$ and zero otherwise.  Thus every number in $\Z + \hf$ is an eigenvalue of $\d$ 
except $\pm\hf$.   Since Equation \eqref{eq:dirac_eigenvalues} implies that
\begin{equation}    
\label{eq:decomposition_into_dirac_eigenspaces}
\V_j \otimes \V_j \otimes \C^2 = F_{n+\hf} \oplus F_{-n +\hf} 
\end{equation}
where $n = 2j+1$, this additional result implies that these two summands have dimensions
$n(n+1)$ and $n(n-1)$, respectively.  Their total dimension is $2n^2$, as we already knew.

\subsection{Digression on quaternions}
\label{subsec:quaternions}

The Hilbert space of bound states for the hydrogen atom, which we have been treating as $L^2(S^3) \otimes \C^2$, also has an elegant quaternionic description.   We can identify $S^3$ with the unit sphere in the quaternions, and $\C^2$ with the quaternions $\HH$ themselves.  Indeed we can identify $\C^2$ with $\HH$ in such a way that multiplication by $-i\sigma_1, -i\sigma_2$ and $-i\sigma_3$ correspond to left multiplication by the quaternions $i,j$ and $k$, respectively, while multiplication by $i$ corresponds to \emph{right} multiplication by the quaternion $i$.   In this way we identify 
$L^2(S^3) \otimes \C^2$ with $L^2(S^3) \otimes \HH$, or equivalently with the space of functions
\[         \psi \maps S^3 \to \HH \]
such that
\[          \int_{S^3} |\psi(q)|^2 < \infty .\]
As we shall see, a dense subspace of functions of this sort extend to functions $\psi \maps \HH - \{0\} \to \HH$ that obey a quaternionic analogue of the Cauchy--Riemann equation.     This lets us apply ideas from quaternionic analysis \cite{FrenkelLibine,Sudbery}. 

For any open set $O \subseteq \HH$, a function $\psi \maps O \to \HH$ is said to be `regular' if it is differentiable in the usual real sense and the quaternionic Cauchy--Riemann equation holds:
\begin{equation}
   \frac{\partial \psi}{\partial q_0} +  i \frac{\partial \psi}{\partial q_1} + j \frac{\partial \psi}{\partial q_2} + k \frac{\partial \psi }{\partial q_3} = 0. 
\end{equation}
Here $q_0, \dots, q_3$ are real coordinates on $\HH$ for which any quaternion $q$ is of the form
\begin{equation}
   q = q_0 + q_1 i + q_2 j + q_3 k .
\end{equation}
Sudbery shows \cite[Thm.\ 1]{Sudbery} that any regular function is infinitely differentiable in the usual real sense, in fact real-analytic.    

Let $U_k$ be the space of regular functions on $\HH - \{0\}$ that are homogeneous of degree $k \in \Z$:
\begin{equation}
 \psi(\alpha q) = \alpha^k f(q)  \qquad \qquad \forall q \in \HH - \{0\}, \alpha \in \R - \{0\}. 
\end{equation}
Clearly any function $\psi \in U_k$ is determined by its restriction to the unit sphere $S^3 \subset \HH$.  In the proof of his Thm.\ 7, Sudbery shows something less obvious:  the restriction is an eigenfunction of the Dirac-like operator $D$ defined in Equation \eqref{eq:Dirac-like_operator}.    To do this, he writes the quaternionic Cauchy--Riemann operator
\begin{equation}
 \overline{\partial} = \frac{\partial}{\partial q_0} +  i \frac{\partial}{\partial q_1} + j \frac{\partial}{\partial q_2} + k \frac{\partial}{\partial q_3} 
\end{equation}
in something like polar coordinates, involving a radial derivative but also the operator $D$.   The radial derivative of the homogeneous function $\psi$ equals $k \psi$ when restricted to $S^3$.  Using this fact and $\overline{\partial} \psi = 0$, he computes $D (\psi|_{S^3})$ and shows that
\begin{equation}
\psi \in U_k \implies D(\psi|_{S^3}) = k \, \psi|_{S^3}
\end{equation}
(although he uses different notation).
Since $\d = D + \tfrac{3}{2}$, this implies that functions in $U_k$ restrict to eigenfunctions of $\d$ with eigenvalue $k + \frac{3}{2}$.   However, in Thm.\ 7 Sudbery also shows that treating $U_k$ as a complex vector space,
\begin{equation}
\dim(U_k) = (k+1)(k+2)
\end{equation}
which by Equation \eqref{eq:dimension_of_dirac_eigenspaces} is exactly the dimension of the
eigenspace of $\d$ with eigenvalue $k + \frac{3}{2}$.   

We thus obtain an isomorphism between $U_k$ and the eigenspace of the operator $\d$ on $L^2(S^3) \otimes \C^2$ with eigenvalue $k + \frac{3}{2}$.   Thus, we get a correspondence between hydrogen atom bound states that are \emph{finite} linear combinations of energy eigenstates and regular functions $\psi \maps \HH - \{0\} \to \HH$ that are finite linear combinations of homogeneous ones.   For infinite linear combinations, the analysis question arises: which infinite sums of homogeneous regular functions $\psi \maps \HH - \{0\} \to \HH$ converge to well-defined regular functions?  

It is worth adding that Frenkel and Libine have also studied bound states of the hydrogen atom using quaternionic analysis \cite[Sec.\ 2.9]{FrenkelLibine}, giving an explicit transform sending homogeneous harmonic functions $\psi \maps \HH - \{0\} \to \C$ to functions in $L^2(\R^3)$ that are eigenvectors of the hydrogen atom Hamiltonian  $- \frac{1}{2} \nabla^2 - \frac{1}{r}$.   Their treatment does not incorporate the spin of the electron, but it could be adapted to do so, by working with regular functions $\psi \maps \HH - \{0\} \to \HH$ (which are automatically harmonic).

\section{The Weyl equation on the Einstein universe }
\label{sec:einstein}

The `Einstein universe' is a name for the manifold $\R \times S^3$ with Lorentzian metric $dt^2 - ds^2$, where $dt^2$ is the usual Riemannian metric on $\R$ and $ds^2$ is the Riemannian metric on the unit sphere.  The Einstein universe has a lot of symmetry: the group $\R \times \SO(4)$ acts as isometries, and the universal cover of $\SO(2,4)$ acts as  conformal transformations.    The former group is the one relevant here, since it acts on the bound states of the hydrogen atom, but it is worth noting that the latter, larger group acts as `dynamical symmetries' of the Kepler problem \cite{Cordani, GuilleminSternberg}: that is, as unitary operators on $L^2(S^3) \otimes \C^2$ that do not all commute with the Hamiltonian.

The Weyl equation is a variant of the Dirac equation that describes massless spin-$\hf$ particles that are chiral, i.e., have an inherent handedness.   We can trivialize the bundle of Weyl spinors over the Einstein universe, using right translations on the group $\R \times \SU(2) \cong \R \times S^3$ to identify every fiber of this bundle with the vector space $\C^2$.   Using this trivialization we can write the left-handed Weyl equation as
\begin{equation}
      \frac{\partial \psi}{\partial t} = -i\d \psi 
\end{equation}
where $\psi \maps \R \times S^3 \to \C^2$ and $\d$ is as defined in the previous section.
The Weyl equation also comes in a right-handed form differing by a sign, $\frac{\partial \psi}{\partial t} = i\d \psi$.   We choose henceforth to work with the left-handed form; this is an arbitrary convention.

We now recall how the relativistic quantum mechanics of a single left-handed massless spin-$\hf$ particle works on the Einstein universe \cite{PaneitzSegal}.    We take $L^2(S^3) \otimes \C^2$ as the Hilbert space and $\d$ as the Hamiltonian.  Since $\d$ is self-adjoint, this Hamiltonian generates a 1-parameter group of unitary operators 
\begin{equation}
      U(t) = \exp(-it \d) . 
\end{equation}
Given any $\psi_0 \in L^2(S^3) \otimes \C^2$, if we let $\psi_t = U(t) \psi_0$
and define a function $\psi \maps \R \times S^3 \to \C^2$ by
\begin{equation}
     \psi(t,x) = \psi_t (x) ,
\end{equation}
then this function will be a distributional solution of the left-handed Weyl equation.

As seen in Equation \eqref{eq:hamiltonian_quantum_3}, the hydrogen atom Hamiltonian $H$ is a function of the Hamiltonian $\d$ for the Weyl equation:
\begin{equation}
    H = - \frac{1}{2 (\d - \hf)^2}.   
\end{equation}
Thus, not only the Hilbert space but also the dynamics of the bound states of the hydrogen atom can be expressed in terms those for the Weyl equation on the Einstein universe.    However, not all the symmetries we have found for the Hamiltonian $H$ are symmetries of $\d$: this is possible because while $H$ is a function of $\d$, $\d$ is not a function of $H$.   Let us make this precise.

In Section \ref{sec:spin} we made $L^2(S^3) \otimes \C^2$ into a unitary representation of $\SU(2)$ in three commuting ways: via left translations, via right translations, and via the spin-$\hf$ representation on $\C^2$.   The self-adjoint generators of these three representations are $A_j$, $B_j$ and $S_j$, respectively.     With the help of Equation \eqref{eq:dirac_operators} we can write the Dirac operator in terms of these:
\begin{equation}
\label{eq:dirac_operator}
\begin{array}{ccl}
\d &=& D + \tfrac{3}{2}  \\  [2pt]
&=& i X_j \otimes \sigma_j + \tfrac{3}{2} \\  [4pt]
&=& 4B_j S_j + \tfrac{3}{2}.
\end{array}
\end{equation}
Using this and the commutation relations listed in Equation \eqref{eq:commutation},  we see $\d$ commutes with the operators $A_j$.  It does not commute with $B_j$ or $S_j$ separately, but it commutes with $B_j + S_j$, since
\begin{equation}
\begin{array}{ccl}
[B_j + S_j, B_k S_k] &=& [B_j, B_k] S_k + B_k [S_j, S_k] \\ [2pt]
&=& i \epsilon_{jk\ell} B_\ell S_k + i \epsilon_{jk\ell} B_k S_\ell  \\ [2pt]
&=& 0.
\end{array}
\end{equation}
It follows that $\d$ commutes with the unitary representation $\rho$ of $\SU(2) \times \SU(2)$ on $L^2(S^3) \otimes \C^2$ whose self-adjoint generators are $A_j $ and $B_j  + S_k$.  Explicitly, this representation is given by
\begin{equation}
\label{eq:new_representation} 
 (\rho(g_1,g_2) \psi)(g) = g_2 \psi(g_1^{-1} g g_2). 
\end{equation}
Geometrically, this representation arises from the natural way to lift the left and right translation actions of $\SU(2)$ on $S^3$ to the spinor bundle of $S^3$.   The asymmetry between left and right here may seem puzzling, but it arose because we arbitrarily chose to trivialize the spinor bundle of $S^3$ using the action of $\SU(2)$ as \emph{left} translations.  Thus, the action of $(g_1,1) \in \SU(2) \times \SU(2)$ on $\psi \in L^2(S^3) \otimes \C^2$ merely left translates $\psi$, while the action of $(1,g_2)$ not only right translates $\psi$ but acts on its value by $g_2$.  We can notice this at the Lie algebra level by noting that the operators $A_j$, which generate left translations, annihilate all the constant spinor-valued functions on $S^3$, which are the elements $\psi \in V_0 \otimes V_0 \otimes \C^2$, while the operators $B_j + S_j$ do not.

Summarizing, this is what we have seen so far.   Made into representations of $\SU(2) \times \SU(2)$ as above, the Hilbert space of bound states of hydrogen atom and the Hilbert space for the left-handed Weyl equation on the Einstein universe are unitarily equivalent.   Moreover, we can express the Hamiltonian for the hydrogen atom in terms of that for the left-handed Weyl equation.

All this is fine mathematics, but there is a physical problem, noticed already by Dirac in a related context: the spectrum of $\d$ is unbounded below, giving states of arbitraily large negative energy.  One widely accepted solution \cite[Sec.\ 6.5]{BSZ} is to modify the complex structure on the Hilbert space, multiplying it by $-1$ on the negative-frequency solutions of the Weyl equation: that is, the subspace of $L^2(S^3) \otimes \C^2$ spanned by eigenvectors $\d \psi = \lambda \psi$ with $\lambda < 0$.   This is an updated version of Dirac's original idea of treating antiparticles as `holes in the sea of negative-energy particles,' or the later idea of switching annihilation and creation operators for negative-frequency solutions.  While typically introduced along with second quantization, the idea of modifying the complex structure also makes sense in the relativistic quantum mechanics of a single particle, which is the context here.

To modify the complex structure on $L^2(S^3) \otimes \C^2$, we use the functional calculus to define an operator 
\begin{equation}     
S = \frac{\d}{|\d|}  
\end{equation}
on this Hilbert space.  This is $1$ on eigenvectors of $\d$ with positive eigenvalue and $-1$ on those with negative eigenvalue; we have seen that $0$ is not an eigenvalue of $\d$, so 
$S$ is well-defined.  If we then define an operator
\begin{equation}       
j = i S ,
\end{equation}
since $S$ is both unitary and self-adjoint, it follows that $j$ is both unitary ($jj^* = j^* j = 1$) and skew-adjoint ($j^* = -j$), and thus a complex structure ($j^2 = -1$).   We henceforth use $\H$ to stand for $L^2(S^3) \otimes \C^2$ made into a complex Hilbert space with the same norm and this new complex structure $j$.  

The operators $\d$ and $|\d|$ are still complex-linear on $\H$, despite the new complex structure, since they commute with $i$ and $S$, and thus $j$.   The operator $\d$ is still self-adjoint on $\H$, since it has an orthonormal basis of eigenvectors with real eigenvalues.   The operator $|\d|$ is not only self-adjoint but positive definite on $\H$, since
\begin{equation}     
\d \psi = \lambda \psi  \; \implies \; |\d| \psi = |\lambda| \psi .
\end{equation}
In fact the operator $|\d|$ generates $U(t)$ as a one-parameter unitary group on $\H$, because 
\begin{equation}      
\exp(-jt|\d|) = \exp(-it\d) = U(t) .
\end{equation}
Thus negative energy states have been eliminated, without changing the time evolution operators $U(t)$ at all, by changing the Hamiltonian from $\d$ to $|\d|$ and
simultaneously changing the complex structure from $i$ to $j$.   

Since all the operators $\rho(g_1,g_2)$ on $L^2(S^3) \otimes \C^2$ given by Equation \eqref{eq:new_representation} commute with $\d$, they also commute with $S$ and thus  with the new complex structure $j = iS$.   Thus $\rho$, which began life as a unitary representation of $\SU(2) \times \SU(2)$ on $L^2(S^3) \otimes \C^2$, gives a complex-linear representation of this group on $\H$, which we call $\rho_\H$.   This representation $\rho_\H$ is unitary, since the norm on $\H$ is the same as that on $L^2(S^3) \otimes \C^2$, and any norm-preserving invertible linear operator on a Hilbert space is unitary.

Furthermore, the representation $\rho_\H$ is unitarily equivalent to $\rho$.  This is a nontrivial fact, because the unitary equivalence between them is not the identity operator.  Indeed the map 
\begin{equation}
\begin{array}{rcl}
I \maps L^2(S^3) \otimes \C^2 &\to  &\H  \\
                              \psi &\mapsto & \psi 
\end{array}
\end{equation}
is not even complex linear: it is complex linear on the $+1$ eigenspace of $S$ but conjugate-linear on the $-1$ eigenspace.   To correct for this, we use a conjugate-linear map 
\begin{equation} 
\begin{array}{c}
C \maps L^2(S^3) \otimes \C^2 \to   L^2(S^3) \otimes \C^2   \\ [4pt]
       (C \psi)(g) = \epsilon \, \overline{\psi}(g)   
\end{array}
\end{equation}
where we regard $\psi$ as a $\C^2$-valued function on $S^3$, let $\overline{\psi}$ denote its componentwise complex conjugate, and multiply $\overline{\psi}$ by 
\begin{equation}
      \epsilon = \left(\begin{array}{cc} 0 & 1 \\ -1 & 0 \end{array} \right) .
\end{equation}
The reason the map $C$ is important is that
 \begin{equation}
\label{eq:C_commutation}
\begin{array}{ccl}
     C(g \psi) = \overline{g} C(\psi) 
\end{array}
\end{equation}
for all $g \in \SU(2)$ and $\psi \in L^2(S^3) \otimes \C^2$.   Checking this is a standard calculation, which we carry out in Appendix \ref{sec:unitary_equivalence}, but conceptually it says that $C$ is an equivalence between the spin-$\hf$ representation of $\SU(2)$ and its conjugate representation.  The desired unitary equivalence between $\rho$ and $\rho_\H$ is then the map 
\begin{equation}
\begin{array}{c}
F \maps L^2(S^3) \otimes \C^2 \to \H \\ [4pt]
    F = I(p_+ + C p_-) 
\end{array}
\end{equation}
where $p_+, p_- \maps L^2(S^3) \otimes \C^2 \to  L^2(S^3) \otimes \C^2$ are the projections of $\psi$ to the $+1$ and $-1$ eigenspaces of $S$, respectively.  We include $I$ here because in the proof of the following theorem we need to keep careful track of the difference between $L^2(S^3) \otimes \C^2$ and $\H$.
 
 \begin{thm} 
 \label{thm:unitary_equivalence_1}
The operator $F \maps L^2(S^3) \otimes \C^2 \to \H$ is a unitary equivalence between the representation $\rho$ of $\SU(2) \times \SU(2)$ on $L^2(S^3) \otimes \C^2$ and the representation $\rho_\H$ of this group on $\H$, and $F \d = \d F$ on the domain of $\d$.
\end{thm}

\begin{proof}
Since the proof is rather lengthy, we defer it to Appendix \ref{sec:unitary_equivalence}.
\end{proof}

Since changing the complex structure on a Hilbert space can be a bit bewildering, let us
summarize the results.  We have a unitary equivalence between the Hilbert space $L^2(S^3) \otimes \C^2$ of bound states of the hydrogen atom and the Hilbert space $\H$ of solutions of the left-handed Weyl equation on $\R \times S^3$ equipped with a complex structure that makes its Hamiltonian positive. The group $\SU(2) \times \SU(2)$ has equivalent unitary representations on these two Hilbert spaces.  The Dirac operator $\d$  acts on both $L^2(S^3) \otimes \C^2$ and $\H$ in an manner compatible with their unitary equivalence.  Finally, both the hydrogen atom Hamiltonian and the Hamiltonian for the left-handed Weyl equation can be expressed in terms of the Dirac operator: the former is
\[    H = - \frac{1}{2 (\d - \hf)^2}   \]
while the latter is just $|\d|$.

\subsection{Digression on the conformal group}
\label{subsec:conformal}

The 15-dimensional Lie group $\widetilde{\SO}(2,4)$ acts as conformal transformations on the Einstein universe.   This action maps null geodesics to null geodesics, and using this we can define an action of $\widetilde{\SO}(2,4)$ as symplectic transformations of the Kepler manifold $T^+ S^3$.   In fact the Kepler manifold is isomorphic to a coadjoint orbit of this group.    Applying geometric quantization to the Kepler manifold we can obtain an irreducible unitary representation of $\widetilde{\SO}(2,4)$ on the Hilbert space $L^2(S^3)$.   This process is subtle and has been carried out in quite a number of ways \cite{BlattnerWolf, CordaniQuantization, GuilleminSternberg, MladenovTsanov}, perhaps because none is fully satisfying in all respects.    For expository accounts we recommend the books by Cordani \cite{Cordani} and by Guillemin and Sternberg \cite{GuilleminSternberg}.

As mentioned in Section \ref{sec:introduction}, we can think of the Kepler manifold as the classical phase space for a free massless spin-$0$ particle in the Einstein universe.   Thus it is not surprising that upon geometric quantization it should give the Hilbert space of states of a massless spin-0 particle in this spacetime.   Indeed, the complex Hilbert space $L^2(S^3)$, which describes bound states of hydrogen ignoring the electron's spin, also serves to describe initial data for solutions of the conformally invariant real Klein--Gordon equation
\begin{equation}
    \frac{\partial^2 \phi}{\partial t^2} = - (\Delta+1)\phi  
\end{equation}
where $\phi \maps \R \times S^3 \to \R$.   The extra $+1$ is necessary for conformal invariance, and we have also seen it in the Hamiltonian for the hydrogen atom when we ignore the electron's spin, in Equation \eqref{eq:hamiltonian_quantum_2}.   As noted by Barut and Kleinert \cite{BarutKleinert} and others, the representation of $\widetilde{\SO}(2,4)$ on $L^2(S^3)$ also can be used to describe `dynamical symmetries' of bound states of hydrogen with a spinless electron.  We saw how this works for the spatial rotation subgroup $\SU(2) \times \SU(2) \subset \widetilde{\SO}(2,4)$ in Section \ref{sec:spinless}.  

The group $\widetilde{\SO}(2,4)$ also has has two irreducible unitary representations on the space $\H$, which describe right- and left-handed spin-$\hf$ particles on the Einstein universe.   These have been studied by many authors; for example Mack and Todorov \cite{MackTodorov} gave three equivalent constructions of these representations, and Mack \cite{Mack} later gave a complete classification of the irreducible positive-energy unitary representations of  $\widetilde{\SO}(2,4)$, which include all three representations just mentioned.   It would be interesting to obtain the two spin-$\hf$ representations of $\widetilde{\SO}(2,4)$ on $\H$ by geometrically quantizing some phase space(s) for a classical spinning massless particle in the Einstein universe.  However, as mentioned, there is still work left to do on clarifying geometric quantization for the spin-0 representation of $\widetilde{\SO}(2,4)$ on $L^2(S^3)$.

\section{Second quantization}
\label{sec:second}

We now `second quantize' the bound states of the hydrogen atom with spin, or equivalently the left-handed Weyl equation on the Einstein universe.   Physically this means that we construct a Hilbert space for arbitrary finite collections of electrons orbiting the nucleus, or arbitrary finite collections of massless left-handed spin-$\hf$ particles in the Einstein universe, while implementing the Pauli exclusion principle.  

To do this, we first recall how to build the fermionic Fock space on an arbitrary Hilbert space $\h$.    We start with the exterior algebra
\begin{equation}
  \Lambda \h  = \bigoplus_{n = 0}^\infty \Lambda^n \h
\end{equation}
and give it the inner product such that if $e_j$ is any orthonormal basis for $\h$, the wedge products $e_{i_1} \wedge \cdots \wedge e_{i_n}$ with $i_1 < \cdots < i_n$ form an orthonormal basis for $\Lambda^n \h$, and the different subspaces $\Lambda^n \h$ are orthogonal.  Completing $\Lambda \h$ with respect to the norm coming from this inner product, we obtain a Hilbert space we call $\Fock \h$.    

If $\h$ is the Hilbert space for a \emph{single} particle of some sort, $\Lambda^n \h$ is the Hilbert space of states of a collection of $n$ particles of this sort, treated as fermions, and $\Fock \h$ is the Hilbert space for arbitrary finite collections of such fermions.  We call $\Fock \h$ the `fermionic Fock space' on $\h$, and $\Lambda^n \h$ the `$n$-particle subspace'.

To define observables on the Fock space, recall how any self-adjoint operator on $\h$ gives rise to one on $\Fock \h$.   First, any unitary operator $U \maps \h \to \h$ gives rise to a unitary operator $\Fock(U) \maps \Fock \h \to \Fock \h$, determined by the property that
\begin{equation}
   \Fock(U) (\psi_1 \wedge \cdots \wedge \psi_n) = U(\psi_1) \wedge \cdots \wedge U(\psi_n)
\end{equation}
for any vectors $\psi_i \in \h$.  Note that 
\begin{equation}  
 \Fock(UV) = \Fock(U) \Fock(V)  .
\end{equation}

Let $\U(\h)$ be the group of unitary operators on $\h$ , and $\U(\Fock \h)$ the group of unitary operators on  $\Fock(\h)$.    If $G$ is any topological group, any strongly continuous unitary representation $R \maps G \to \U(\h)$ gives rise to a strongly continuous representation $\Fock R \maps G \to \U(\Fock \h)$, defined by
\begin{equation}
(\Fock R)(g) = \Fock (R(g)) .
\end{equation}
In particular, any self-adjoint operator $A$ on $\h$ gives a strongly continuous one-parameter unitary group $\exp(itA)$ on $\h$, and thus a strongly contiuous one-parameter unitary group $\Fock(\exp(itA))$ on $\Fock\h$.  Stone's theorem says the latter is generated by a unique self-adjoint operator on $\Fock\h$, which we call $d\Fock(A)$.  We thus have
\begin{equation}  
\label{eq:fock_exponentiation}
\exp(it d\Fock(A)) = \Fock(\exp(itA)) 
\end{equation}
for all $t \in \R$.  If the vectors $\psi_1, \dots, \psi_n$ are in the domain of $A$, we can differentiate both sides of the above formula applied to $\psi_1 \wedge \cdots \wedge \psi_n$ and set $t = 0$, obtaining
\begin{equation}
    d\Fock(A) (\psi_1 \wedge \cdots \wedge \psi_n) = \sum_{i = 1}^n \psi_1 \wedge \cdots \wedge A \psi_i \wedge \cdots \wedge \psi_n .
 \end{equation}
In particular if all $\psi_i$ are eigenvectors of $A$, then their wedge product is an eigenvector of $d\Fock(A)$:
\begin{equation}
A \psi_i = \lambda_i \psi_i \implies   d \Fock(A)(\psi_1 \wedge \cdots \wedge \psi_n) = (\lambda_1 + \cdots + \lambda_n) (\psi_1 \wedge \cdots \wedge \psi_n). 
\end{equation}

We can apply all this mathematics by taking $\h$ to be the Hilbert space of bound states of a hydrogen atom: 
\begin{equation}
   \h = L^2(S^3) \otimes \C^2
\end{equation}
Then $\Fock \h$ is the Hilbert space for an arbitrary finite collection of electrons occupying such states.     In particular, if $H$ is the hydrogen atom Hamiltonian, then $d\Fock(H)$ restricted to $n$-particle subspace $\Lambda^n \h$ is the Hamiltonian for an idealized atom with $n$ noninteracting electrons.   Since in fact electrons \emph{do} interact, the lowest-energy eigenstate in the $n$-particle space gives a very crude approximation to the $n$th element in the periodic table.  To do better we must modify the Hamiltonian.  We  discuss this in the next section.   

Alternatively we can start with $\H$, the Hilbert space of a single left-handed massless spin-$\hf$ partice in the Einstein universe.   We have seen that the Hamiltonian for such a particle is $|\d|$.  Then $\Fock \H$ is the Hilbert space for an arbitrary collection of left-handed massless spin-$\hf$ particles, treated as fermions.   If these particles are noninteracting, their Hamiltonian is $d\Fock(|\d|)$, and we have a free quantum field theory.
 
A unitary map between Hilbert spaces, for example the unitary operator $F \maps \h \to \H$ of Theorem \ref{thm:unitary_equivalence_1}, induces a unitary operator between their fermionic Fock spaces
\begin{equation}
\begin{array}{c}    \Fock(F) \maps \Fock \h  \to \Fock \H \\  [3pt]
    \Fock(F) (\psi_1 \wedge \cdots \wedge \psi_n) = F \psi_1 \wedge \cdots \wedge F \psi_n 
\end{array}
\end{equation}
for all $\psi_1, \dots, \psi_n \in \h$.    Thus, an equivalence between two theories at the single-particle level induces an equivalence between their second quantized versions.  

\begin{thm}
 \label{thm:unitary_equivalence_2}
The map $\Fock(F) \maps \Fock \h \to \Fock \H$ is a unitary equivalence between the representation $\Fock(\rho)$ of the group $\SU(2) \times \SU(2)$ on the fermionic Fock space  $\Fock \h$ for bound states of the hydrogen atom Hamiltonian and the representation $\Fock(\rho_\H)$ of this group on the fermionic Fock space $\Fock \H$ for left-handed massless spin-$\hf$ particles on the Einstein universe.   That is,
\[      \Fock(F) \Fock(\rho(g_1,g_2)) = \Fock(\rho_\H(g_1,g_2)) \Fock(F)  \]
for all $(g_1,g_2) \in \SU(2) \times \SU(2)$.   Moreover 
\[   \Fock(F) d\Fock(\d) = d\Fock(\d) \Fock(F)  \]
on the domain of $d\Fock(\d)$.
\end{thm}

\begin{proof}
We defer the proof to Appendix \ref{sec:unitary_equivalence}.
\end{proof}

\section{Multi-electron atoms --- the Madelung rules}
\label{sec:aufbau}

Having set up an isomorphism between the Hilbert space for collections of electrons orbiting a nucleus and the Hilbert space of a massless spin-$\hf$ quantum field on the Einstein universe, we can now try to exploit it.   In what follows, we describe a Hamiltonian for a spin-$\hf$ quantum field on the Einstein universe whose lowest-energy $N$-particle state roughly corresponds to the ground state of an $N$-electron atom.   This is not to be taken seriously as atomic physics, for various reasons that we shall explain.     It is just an amusing mathematical exercise: to exhibit a simple quantum field theory whose structure loosely mimics that of the periodic table.

We have seen that in the nonrelativistic limit, the hydrogen atom with spin has the Hilbert space
of bound states $\h = L^2(S^3) \otimes \C^2$.  In atomic physics,  the eigenspaces of the Hamiltonian $H$ on $\h$ are called `shells', while the joint eigenspaces of the operators $H$ and the angular momentum squared, $L^2$, are called `subshells'.   We denote the shells as 
\begin{equation}
  \h_n = \{ \psi \in \h: \; H\psi =  -\frac{1}{2n^2} \psi \} 
\end{equation}
and the subshells as
\begin{equation}
    \h_{n,\ell} = \{ \psi \in \h_n : \;  L^2 \psi = \ell(\ell + 1) \psi \} .
\end{equation}
The shells are direct sums of subshells as follows:
\begin{equation}
   \h_n = \bigoplus_{\ell = 0}^{n-1} \h_{n,\ell}  
 \end{equation}
and the direct sum of all the shells is $\h$:
\begin{equation}
    \h = \bigoplus_{n = 1}^\infty  \h_n  .
 \end{equation}
The dimension of the subshell $\h_{n,\ell}$ is $2(2\ell + 1)$, so the dimension of the shell $\h_n$ is
\[    2(1 + 3 + 5 + \cdots + (2n - 1)) = 2n^2  .\]

The `Aufbau principle' is an approximate way to describe the ground state of an $N$-electron atom as a state $\phi$ in the $N$-particle subspace of the Fock space $\Fock \h$.   To do this, we choose a Hamiltonian $H_{\text{Fock}}$ on $\Fock \h$ and decree that $\phi$ must minimize the expected energy $\langle \phi , H_{\text{Fock}}\, \phi \rangle$ among all unit vectors
in the $N$-particle subspace.   However, we choose the Hamiltonian $H_{\text{Fock}}$ in a very simplistic way.  We ignore the details of electron-electron interactions.  Instead, we simply assign an energy $E_{n,\ell}$ to each subshell, let $H_{\text{single}}$ be the unique Hamiltonian on $\h$ such that 
\begin{equation}
      \psi \in \h_{n,\ell} \implies   H_{\text{single}} \psi = E_{n,\ell} \,\psi  ,
\end{equation}
and then we let 
\begin{equation}     
    H_{\text{Fock}} = d\Fock(H_{\text{single}})  .
\end{equation}
Thus, $H_{\text{Fock}}$ has a basis of eigenvectors that are wedge products of single-particle states lying in various subshells.    Explicitly, if $\phi = \psi_1 \wedge \cdots \wedge \psi_N$ where $\psi_i \in \h_{n_i, \ell_i}$, we have
\begin{equation}
  H_{\text{Fock}} \, \phi = (E_{n_1,\ell_1} + \cdots + E_{n_N, \ell_N}) \phi  .
\end{equation}
Thus, we can minimize $\langle \phi , H_{\text{Fock}}\, \phi \rangle$ among unit vectors
in the $N$-particle subspace by choosing $N$ distinct basis vectors $\psi_i = |n_i, \ell_i, m_i, s_i \rangle $ in a way that minimizes the total energy $E_{n_1,\ell_1} + \cdots + E_{n_N, \ell_N}$.

If we follow this recipe taking $H_{\text{single}}$ to be the hydrogen atom Hamiltonian $H$, 
we get results that do not closely match the observed periodic table of elements.  With this choice, 
 $E_{n,\ell} = -1/2n^2$, which depends only on the shell, not the subshell.   Thus, this choice makes no prediction about the order in which subshells are filled.

For the recipe to give results that more closely match the periodic table, we need to choose the energies $E_{n,\ell}$ in a more clever way.   In 1936, Madelung \cite{Madelung} argued for these rules:
\begin{itemize}
\item subshells are filled in order of increasing value of $n + \ell$;
\item for subshells with the same value of $n + \ell$, subshells are filled in order of decreasing $\ell$ (or equivalently, increasing $n$). 
\end{itemize} 
The pattern of subshell filling is then shown in Figure \ref{figure_1}.   

\begin{figure}[h!]
\includegraphics[width = 27em]{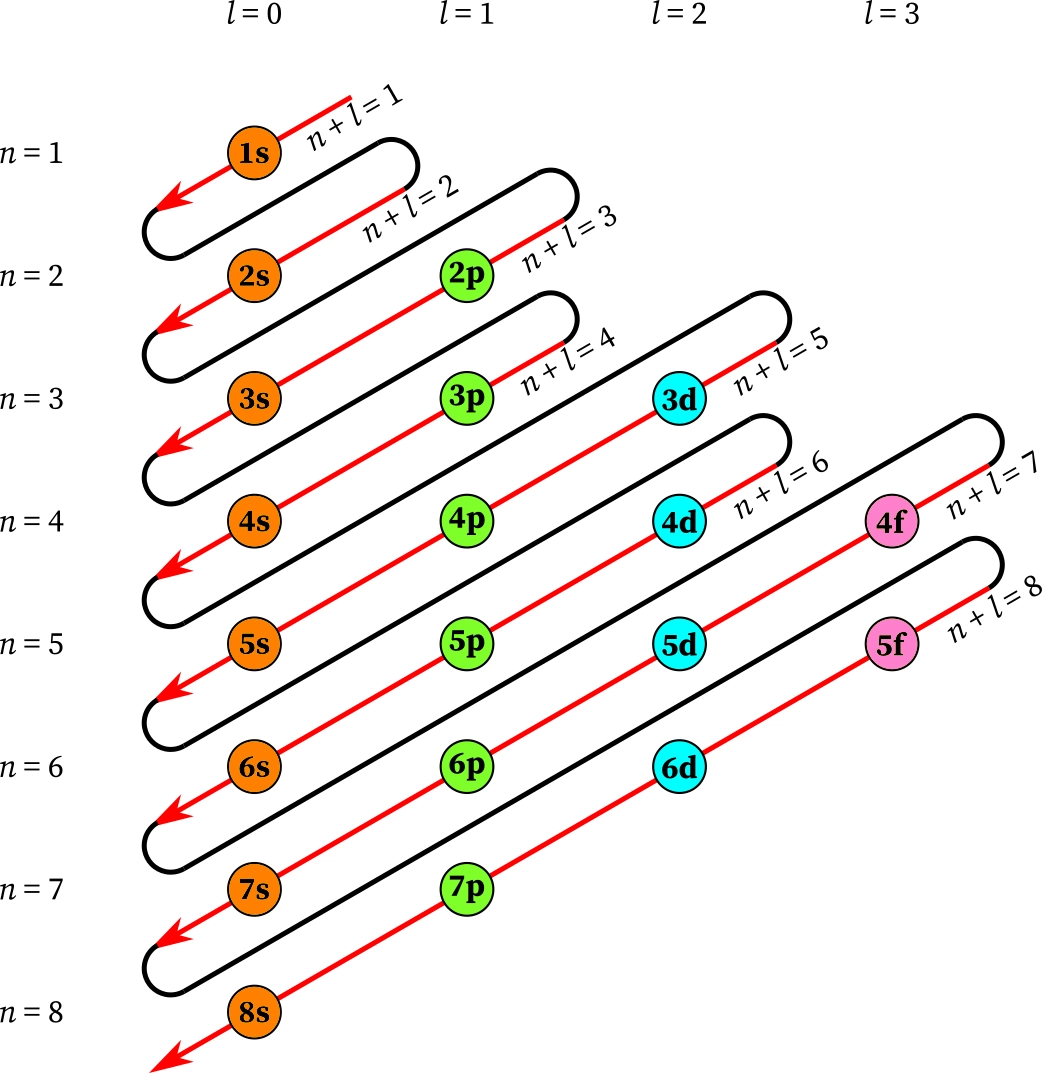} 
\vskip 1em
\caption{The Madelung rules.}
\label{figure_1}
\end{figure}

In 1945, Wiswesser \cite{Wiswesser} noted that the Madelung rules follow from the recipe we outlined if we choose $E_{n,\ell} = n + \ell - \frac{\ell}{\ell + 1}$.
There are many other functions of $n$ and $\ell$ that achieve the same effect.  For example, we can also obtain the Madelung rules if we take
\begin{equation}       
\label{eq:interaction_hamiltonian}
 E_{n,\ell} = 2n + (2\ell + 1) + (2\ell + 1)^{-1} ,
\end{equation}
and this formula is more convenient for us.   In Figure \ref{figure_2} and Table \ref{table_1} we show how the Madelung rules and this formula play out for the first 120 elements.

The Madelung rules do not always hold.   The first exception is element 24, chromium. The Madelung rules predict that chromium has 2 electrons in the $4s$ subshell and 4 electrons in the $3d$, while in fact it has 1 in the $4s$ and 5 in the $3d$.   There are eleven exceptions to the Madelung rules in the so-called $d$-block elements (shown in purple in Figure \ref{figure_2}), and nine exceptions in the $f$-block elements (shown in pink).   Nonetheless the general structure of the periodic table is in reasonably good accord with the Madelung rules for all the elements studied chemically so far, though relativistic effects may end this for very heavy elements.

Thus it is of some interest, if only as a curiosity, to define a Hamiltonian on the Hilbert space $\h$ that takes the eigenvalue $E_{n,\ell}$ in the subshell $\h_{n,\ell}$.  These energies are not at all close to the actual energies of the various multi-electron atoms, and any monotone function of $E_{n,\ell}$ would also give the Madelung rules---but this particular Hamiltonian is fairly simple.  More realistic single-particle Hamiltonians for electrons in atoms are known \cite{Latter}, but they are not connected to the mathematics we are studying here.

Recall from Equation \eqref{eq:hydrogen} that
\[
  \begin{array}{cclll}
    A^2 |n , \ell, m,  s \rangle  &=& j(j+1) |n , \ell, m, s \rangle &=& \fth(n^2 - 1) |n , \ell, m, s \rangle \\ [3pt]
     L^2 |n , \ell, m, s \rangle &=& \ell(\ell + 1) |n , \ell, m, s \rangle.
     \end{array}
\]
where $n = 2j + 1$ as usual.  The Duflo isomorphism, as discussed in Section \ref{subsec:duflo}, makes it natural to define operators
\begin{equation}
   \tilde{A}^2 = A^2 + \fth, \qquad \tilde{L}^2 = L^2 + \fth .
\end{equation}
If we then define $\tilde{A}$ and $\tilde{L}$ to be the square roots of these operators, we have
\begin{equation}
  \begin{array}{cclcl}
    \tilde{A} |n , \ell, m, s \rangle  &=& \hf n |n , \ell, m, s \rangle    \\ [3pt]
     \tilde{L} |n , \ell, m, s \rangle &=& (\ell + \hf) |n , \ell, m, s \rangle
     \end{array}
\end{equation}
and thus by Equation \eqref{eq:interaction_hamiltonian} 
\begin{equation}
   (2\tilde{A} + 2\tilde{L} + (2\tilde{L})^{-1})|n , \ell, m, s \rangle =
 E_{n,\ell} |n , \ell, m,  s\rangle .
 \end{equation}
 
 This suggests taking our single-particle Hamiltonian to be
 \begin{equation}
 H_{\text{single}} = 2\tilde{A} + 2\tilde{L} + (2\tilde{L})^{-1}.
\end{equation}
If we then define a Hamiltonian on the Fock space $\Fock \h$ by
\begin{equation}
   H_{\text{Fock}} = d\Fock(H_{\text{single}}) 
\end{equation}
and create an orthonormal basis $\psi_i$ of eigenvectors of $H_{\text{Fock}}$, listed in order of increasing eigenvalue, these eigenvectors correspond to the elements with subshells filled as predicted by the Madelung rules.  The one exception is the state $\psi_0$ with no electrons, sometimes called `element zero' and identified with the neutron.  For example:
\[
\begin{array}{ccll} 
\psi_1 &=& |1,0,0,\hf \rangle & \text{hydrogen} \\ \\
\psi_2 &=& |1,0,0,\hf \rangle \wedge |1,0,0,-\hf\rangle & \text{helium} \\ \\
\psi_3 &=& |1,0,0,\hf \rangle \wedge |1,0,0,-\hf\rangle 
\wedge |2,0,0,\hf\rangle & \text{lithium} \\ \\
\psi_4 &=& |1,0,0,\hf \rangle \wedge |1,0,0,-\hf\rangle 
\wedge |2,0,0,\hf\rangle \wedge |2,0,0,-\hf\rangle & \text{beryllium} \\ \\
\psi_5 &=& |1,0,0,\hf \rangle \wedge |1,0,0,-\hf\rangle 
\wedge |2,0,0,\hf\rangle \wedge |2,0,0,-\hf\rangle 
\wedge |2,1,1,\hf\rangle & \text{boron}
\end{array}
\]
Here the assignments of magnetic quantum numbers $m$ and spins $s$ are not determined by the rules we have laid out.   These are governed, at least approximately, by `Hund's rules': 
\begin{itemize}
\item  every $m$ state in a subshell is singly occupied before any is doubly occupied;
\item  all of the electrons in singly occupied orbitals have the same spin.
\end{itemize}
We could go further and attempt to choose a simple Hamiltonian for which the principle of energy mimization also gives Hund's rules.   However, we prefer to stop here, leaving the reader with the challenge of finding more substantial connections between atomic physics and quantum field theory on the Einstein universe.

\vskip 1em
\begin{figure}[h!]
\newcommand{\CommonElementTextFormat}[4]
{
  \begin{minipage}{2cm}
    \centering
      {\textbf{#1}  #2}%
      \linebreak \linebreak
      {\textbf{#3}}%
      \linebreak \linebreak
      {{#4}}
  \end{minipage}
}

\newcommand{\NaturalElementTextFormat}[4]{\CommonElementTextFormat{#1}{}{\LARGE {#3}}{#4}}

\begin{tikzpicture}[font=\sffamily, scale=0.3, transform shape]

%% Fill Color Styles
  \tikzstyle{ElementFill} = [fill=yellow!15]
  \tikzstyle{AlkaliMetalFill} = [fill=orange]
  \tikzstyle{MetalFill} = [fill=lightblue]
  \tikzstyle{NonmetalFill} = [fill=green]
%  \tikzstyle{HalogenFill} = [fill=green!40]
%  \tikzstyle{NobleGasFill} = [fill=green!55]
  \tikzstyle{LanthanideActinideFill} = [fill=pink]

%% Element Styles
  \tikzstyle{Element} = [draw=black, ElementFill,
    minimum width=2.75cm, minimum height=2.75cm, node distance=2.75cm]
  \tikzstyle{AlkaliMetal} = [Element, AlkaliMetalFill]
  \tikzstyle{AlkalineEarthMetal} = [Element, AlkalineEarthMetalFill]
  \tikzstyle{Metal} = [Element, MetalFill]
  \tikzstyle{Metalloid} = [Element, MetalloidFill]
  \tikzstyle{Nonmetal} = [Element, NonmetalFill]
  \tikzstyle{Halogen} = [Element, HalogenFill]
  \tikzstyle{NobleGas} = [Element, NobleGasFill]
  \tikzstyle{LanthanideActinide} = [Element, LanthanideActinideFill]
  \tikzstyle{PeriodLabel} = [font={\sffamily\LARGE}, node distance=2.8cm]
  \tikzstyle{GroupLabel} = [font={\sffamily\LARGE}, minimum width=2.75cm, node distance=2.0cm]
  \tikzstyle{TitleLabel} = [font={\sffamily\Huge\bfseries}]

%% Group 1 - IA
  \node[name=H, AlkaliMetal] {\NaturalElementTextFormat{1}{1.0079}{H}{Hydrogen}};
  \node[name=Li, below of=H, AlkaliMetal] {\NaturalElementTextFormat{3}{6.941}{Li}{Lithium}};
  \node[name=Na, below of=Li, AlkaliMetal] {\NaturalElementTextFormat{11}{22.990}{Na}{Sodium}};
  \node[name=K, below of=Na, AlkaliMetal] {\NaturalElementTextFormat{19}{39.098}{K}{Potassium}};
  \node[name=Rb, below of=K, AlkaliMetal] {\NaturalElementTextFormat{37}{85.468}{Rb}{Rubidium}};
  \node[name=Cs, below of=Rb, AlkaliMetal] {\NaturalElementTextFormat{55}{132.91}{Cs}{Caesium}};
  \node[name=Fr, below of=Cs, AlkaliMetal] {\NaturalElementTextFormat{87}{223}{Fr}{Francium}};

%% Group 2 - IIA
  \node[name=Be, right of=Li, AlkaliMetal] {\NaturalElementTextFormat{4}{9.0122}{Be}{Beryllium}};
    \node[name=He, above of=Be, AlkaliMetal] {\NaturalElementTextFormat{2}{4.0025}{He}{Helium}};
  \node[name=Mg, below of=Be, AlkaliMetal] {\NaturalElementTextFormat{12}{24.305}{Mg}{Magnesium}};
  \node[name=Ca, below of=Mg, AlkaliMetal] {\NaturalElementTextFormat{20}{40.078}{Ca}{Calcium}};
  \node[name=Sr, below of=Ca, AlkaliMetal] {\NaturalElementTextFormat{38}{87.62}{Sr}{Strontium}};
  \node[name=Ba, below of=Sr, AlkaliMetal] {\NaturalElementTextFormat{56}{137.33}{Ba}{Barium}};
  \node[name=Ra, below of=Ba, AlkaliMetal] {\NaturalElementTextFormat{88}{226}{Ra}{Radium}};

%% Group 3 - IIIB
  \node[name=Sc, right of=Ca, Metal] {\NaturalElementTextFormat{21}{44.956}{Sc}{Scandium}};
  \node[name=Y, below of=Sc, Metal] {\NaturalElementTextFormat{39}{88.906}{Y}{Yttrium}};
  \node[name=LaLu, below of=Y, Metal] {\NaturalElementTextFormat{71}{}{Lu}{Lutetium}};
  \node[name=AcLr, below of=LaLu, Metal] {\NaturalElementTextFormat{103}{}{Lw}{Lawrencium}};

%% Group 4 - IVB
  \node[name=Ti, right of=Sc, Metal] {\NaturalElementTextFormat{22}{47.867}{Ti}{Titanium}};
  \node[name=Zr, below of=Ti, Metal] {\NaturalElementTextFormat{40}{91.224}{Zr}{Zirconium}};
  \node[name=Hf, below of=Zr, Metal] {\NaturalElementTextFormat{72}{178.49}{Hf}{Halfnium}};
  \node[name=Rf, below of=Hf, Metal] {\NaturalElementTextFormat{104}{261}{Rf}{Rutherfordium}};

%% Group 5 - VB
  \node[name=V, right of=Ti, Metal] {\NaturalElementTextFormat{23}{50.942}{V}{Vanadium}};
  \node[name=Nb, below of=V, Metal] {\NaturalElementTextFormat{41}{92.906}{Nb}{Niobium}};
  \node[name=Ta, below of=Nb, Metal] {\NaturalElementTextFormat{73}{180.95}{Ta}{Tantalum}};
  \node[name=Db, below of=Ta, Metal] {\NaturalElementTextFormat{105}{262}{Db}{Dubnium}};

%% Group 6 - VIB
  \node[name=Cr, right of=V, Metal] {\NaturalElementTextFormat{24}{51.996}{Cr}{Chromium}};
  \node[name=Mo, below of=Cr, Metal] {\NaturalElementTextFormat{42}{95.94}{Mo}{Molybdenum}};
  \node[name=W, below of=Mo, Metal] {\NaturalElementTextFormat{74}{183.84}{W}{Tungsten}};
  \node[name=Sg, below of=W, Metal] {\NaturalElementTextFormat{106}{266}{Sg}{Seaborgium}};

%% Group 7 - VIIB
  \node[name=Mn, right of=Cr, Metal] {\NaturalElementTextFormat{25}{54.938}{Mn}{Manganese}};
  \node[name=Tc, below of=Mn, Metal] {\NaturalElementTextFormat{43}{96}{Tc}{Technetium}};
  \node[name=Re, below of=Tc, Metal] {\NaturalElementTextFormat{75}{186.21}{Re}{Rhenium}};
  \node[name=Bh, below of=Re, Metal] {\NaturalElementTextFormat{107}{264}{Bh}{Bohrium}};

%% Group 8 - VIIIB
  \node[name=Fe, right of=Mn, Metal] {\NaturalElementTextFormat{26}{55.845}{Fe}{Iron}};
  \node[name=Ru, below of=Fe, Metal] {\NaturalElementTextFormat{44}{101.07}{Ru}{Ruthenium}};
  \node[name=Os, below of=Ru, Metal] {\NaturalElementTextFormat{76}{190.23}{Os}{Osmium}};
  \node[name=Hs, below of=Os, Metal] {\NaturalElementTextFormat{108}{277}{Hs}{Hassium}};

%% Group 9 - VIIIB
  \node[name=Co, right of=Fe, Metal] {\NaturalElementTextFormat{27}{58.933}{Co}{Cobalt}};
  \node[name=Rh, below of=Co, Metal] {\NaturalElementTextFormat{45}{102.91}{Rh}{Rhodium}};
  \node[name=Ir, below of=Rh, Metal] {\NaturalElementTextFormat{77}{192.22}{Ir}{Iridium}};
  \node[name=Mt, below of=Ir, Metal] {\NaturalElementTextFormat{109}{268}{Mt}{Meitnerium}};

%% Group 10 - VIIIB
  \node[name=Ni, right of=Co, Metal] {\NaturalElementTextFormat{28}{58.693}{Ni}{Nickel}};
  \node[name=Pd, below of=Ni, Metal] {\NaturalElementTextFormat{46}{106.42}{Pd}{Palladium}};
  \node[name=Pt, below of=Pd, Metal] {\NaturalElementTextFormat{78}{195.08}{Pt}{Platinum}};
  \node[name=Ds, below of=Pt, Metal] {\NaturalElementTextFormat{110}{281}{Ds}{Darmstadtium}};

%% Group 11 - IB
  \node[name=Cu, right of=Ni, Metal] {\NaturalElementTextFormat{29}{63.546}{Cu}{Copper}};
  \node[name=Ag, below of=Cu, Metal] {\NaturalElementTextFormat{47}{107.87}{Ag}{Silver}};
  \node[name=Au, below of=Ag, Metal] {\NaturalElementTextFormat{79}{196.97}{Au}{Gold}};
  \node[name=Rg, below of=Au, Metal] {\NaturalElementTextFormat{111}{280}{Rg}{Roentgenium}};

%% Group 12 - IIB
  \node[name=Zn, right of=Cu, Metal] {\NaturalElementTextFormat{30}{65.39}{Zn}{Zinc}};
  \node[name=Cd, below of=Zn, Metal] {\NaturalElementTextFormat{48}{112.41}{Cd}{Cadmium}};
  \node[name=Hg, below of=Cd, Metal] {\NaturalElementTextFormat{80}{200.59}{Hg}{Mercury}};
  \node[name=Cn, below of=Hg, Metal] {\NaturalElementTextFormat{112}{285}{Cn}{Copernicium}};

%% Group 13 - IIIA
  \node[name=Ga, right of=Zn, Nonmetal] {\NaturalElementTextFormat{31}{69.723}{Ga}{Gallium}};
  \node[name=Al, above of=Ga, Nonmetal] {\NaturalElementTextFormat{13}{26.982}{Al}{Aluminium}};
  \node[name=B, above of=Al, Nonmetal] {\NaturalElementTextFormat{5}{10.811}{B}{Boron}};
  \node[name=In, below of=Ga, Nonmetal] {\NaturalElementTextFormat{49}{114.82}{In}{Indium}};
  \node[name=Tl, below of=In, Nonmetal] {\NaturalElementTextFormat{81}{204.38}{Tl}{Thallium}};
  \node[name=Nh, below of=Tl, Nonmetal] {\NaturalElementTextFormat{113}{284}{Nh}{Nihonium}};

%% Group 14 - IVA
  \node[name=C, right of=B, Nonmetal] {\NaturalElementTextFormat{6}{12.011}{C}{Carbon}};
  \node[name=Si, below of=C, Nonmetal] {\NaturalElementTextFormat{14}{28.086}{Si}{Silicon}};
  \node[name=Ge, below of=Si, Nonmetal] {\NaturalElementTextFormat{32}{72.64}{Ge}{Germanium}};
  \node[name=Sn, below of=Ge, Nonmetal] {\NaturalElementTextFormat{50}{118.71}{Sn}{Tin}};
  \node[name=Pb, below of=Sn, Nonmetal] {\NaturalElementTextFormat{82}{207.2}{Pb}{Lead}};
  \node[name=Fl, below of=Pb, Nonmetal] {\NaturalElementTextFormat{114}{289}{Fl}{Flerovium}};

%% Group 15 - VA
  \node[name=N, right of=C, Nonmetal] {\NaturalElementTextFormat{7}{14.007}{N}{Nitrogen}};
  \node[name=P, below of=N, Nonmetal] {\NaturalElementTextFormat{15}{30.974}{P}{Phosphorus}};
  \node[name=As, below of=P, Nonmetal] {\NaturalElementTextFormat{33}{74.922}{As}{Arsenic}};
  \node[name=Sb, below of=As, Nonmetal] {\NaturalElementTextFormat{51}{121.76}{Sb}{Antimony}};
  \node[name=Bi, below of=Sb, Nonmetal] {\NaturalElementTextFormat{83}{208.98}{Bi}{Bismuth}};
  \node[name=Mc, below of=Bi, Nonmetal] {\NaturalElementTextFormat{115}{288}{Mc}{Moscovium}};

%% Group 16 - VIA
  \node[name=O, right of=N, Nonmetal] {\NaturalElementTextFormat{8}{15.999}{O}{Oxygen}};
  \node[name=S, below of=O, Nonmetal] {\NaturalElementTextFormat{16}{32.065}{S}{Sulphur}};
  \node[name=Se, below of=S, Nonmetal] {\NaturalElementTextFormat{34}{78.96}{Se}{Selenium}};
  \node[name=Te, below of=Se, Nonmetal] {\NaturalElementTextFormat{52}{127.6}{Te}{Tellurium}};
  \node[name=Po, below of=Te, Nonmetal] {\NaturalElementTextFormat{84}{209}{Po}{Polonium}};
  \node[name=Lv, below of=Po, Nonmetal] {\NaturalElementTextFormat{116}{293}{Lv}{Livermorium}};

%% Group 17 - VIIA
  \node[name=F, right of=O, Nonmetal] {\NaturalElementTextFormat{9}{18.998}{F}{Fluorine}};
  \node[name=Cl, below of=F, Nonmetal] {\NaturalElementTextFormat{17}{35.453}{Cl}{Chlorine}};
  \node[name=Br, below of=Cl, Nonmetal] {\NaturalElementTextFormat{35}{79.904}{Br}{Bromine}};
  \node[name=I, below of=Br, Nonmetal] {\NaturalElementTextFormat{53}{126.9}{I}{Iodine}};
  \node[name=At, below of=I, Nonmetal] {\NaturalElementTextFormat{85}{210}{At}{Astatine}};
  \node[name=Ts, below of=At, Nonmetal] {\NaturalElementTextFormat{117}{292}{Ts}{Tennessine}}; 

%% Group 18 - VIIIA
  \node[name=Ne, right of=F, Nonmetal] {\NaturalElementTextFormat{10}{20.180}{Ne}{Neon}};
  \node[name=Ar, below of=Ne, Nonmetal] {\NaturalElementTextFormat{18}{39.948}{Ar}{Argon}};
  \node[name=Kr, below of=Ar, Nonmetal] {\NaturalElementTextFormat{36}{83.8}{Kr}{Krypton}};
  \node[name=Xe, below of=Kr, Nonmetal] {\NaturalElementTextFormat{54}{131.29}{Xe}{Xenon}};
  \node[name=Rn, below of=Xe, Nonmetal] {\NaturalElementTextFormat{86}{222}{Rn}{Radon}};
  \node[name=Og, below of=Rn, Nonmetal] {\NaturalElementTextFormat{118}{294}{Og}{Oganesson}}; 

%% Period
  \node[name=Period1, left of=H, PeriodLabel] {\Huge {\boldmath $n = 1$}};
  \node[name=Period2, left of=Li, PeriodLabel] {\Huge {\boldmath $n = 2$}};
  \node[name=Period3, left of=Na, PeriodLabel] {\Huge {\boldmath $n = 3$}}; 
  \node[name=Period4, left of=K, PeriodLabel]  {\Huge {\boldmath $n = 4$}}; 
  \node[name=Period5, left of=Rb, PeriodLabel] {\Huge {\boldmath $n = 5$}};
  \node[name=Period6, left of=Cs, PeriodLabel] {\Huge {\boldmath $n = 6$}};
  \node[name=Period7, left of=Fr, PeriodLabel] {\Huge {\boldmath $n = 7$}};

\iffalse
%% Group
  \node[name=Group1, above of=H, GroupLabel] {1 \hfill IA};
  \node[name=Group2, above of=Be, GroupLabel] {2 \hfill IIA};
  \node[name=Group3, above of=Sc, GroupLabel] {3 \hfill IIIA};
  \node[name=Group4, above of=Ti, GroupLabel] {4 \hfill IVB};
  \node[name=Group5, above of=V, GroupLabel] {5 \hfill VB};
  \node[name=Group6, above of=Cr, GroupLabel] {6 \hfill VIB};
  \node[name=Group7, above of=Mn, GroupLabel] {7 \hfill VIIB};
  \node[name=Group8, above of=Fe, GroupLabel] {8 \hfill VIIIB};
  \node[name=Group9, above of=Co, GroupLabel] {9 \hfill VIIIB};
  \node[name=Group10, above of=Ni, GroupLabel] {10 \hfill VIIIB};
  \node[name=Group11, above of=Cu, GroupLabel] {11 \hfill IB};
  \node[name=Group12, above of=Zn, GroupLabel] {12 \hfill IIB};
  \node[name=Group13, above of=B, GroupLabel] {13 \hfill IIIA};
  \node[name=Group14, above of=C, GroupLabel] {14 \hfill IVA};
  \node[name=Group15, above of=N, GroupLabel] {15 \hfill VA};
  \node[name=Group16, above of=O, GroupLabel] {16 \hfill VIA};
  \node[name=Group17, above of=F, GroupLabel] {17 \hfill VIIA};
  \node[name=Group18, above of=He, GroupLabel] {18 \hfill VIIIA};
\fi

%% Lanthanide
  \node[name=La, below of=Rf, LanthanideActinide, yshift=-1cm] {\NaturalElementTextFormat{57}{138.91}{La}{Lanthanum}};
  \node[name=Ce, right of=La, LanthanideActinide] {\NaturalElementTextFormat{58}{140.12}{Ce}{Cerium}};
  \node[name=Pr, right of=Ce, LanthanideActinide] {\NaturalElementTextFormat{59}{140.91}{Pr}{Praseodymium}};
  \node[name=Nd, right of=Pr, LanthanideActinide] {\NaturalElementTextFormat{60}{144.24}{Nd}{Neodymium}};
  \node[name=Pm, right of=Nd, LanthanideActinide] {\NaturalElementTextFormat{61}{145}{Pm}{Promethium}};
  \node[name=Sm, right of=Pm, LanthanideActinide] {\NaturalElementTextFormat{62}{150.36}{Sm}{Samarium}};
  \node[name=Eu, right of=Sm, LanthanideActinide] {\NaturalElementTextFormat{63}{151.96}{Eu}{Europium}};
  \node[name=Gd, right of=Eu, LanthanideActinide] {\NaturalElementTextFormat{64}{157.25}{Gd}{Gadolinium}};
  \node[name=Tb, right of=Gd, LanthanideActinide] {\NaturalElementTextFormat{65}{158.93}{Tb}{Terbium}};
  \node[name=Dy, right of=Tb, LanthanideActinide] {\NaturalElementTextFormat{66}{162.50}{Dy}{Dysprosium}};
  \node[name=Ho, right of=Dy, LanthanideActinide] {\NaturalElementTextFormat{67}{164.93}{Ho}{Holmium}};
  \node[name=Er, right of=Ho, LanthanideActinide] {\NaturalElementTextFormat{68}{167.26}{Er}{Erbium}};
  \node[name=Tm, right of=Er, LanthanideActinide] {\NaturalElementTextFormat{69}{168.93}{Tm}{Thulium}};
  \node[name=Yb, right of=Tm, LanthanideActinide] {\NaturalElementTextFormat{70}{173.04}{Yb}{Ytterbium}};
%  \node[name=Lu, right of=Yb, LanthanideActinide] {\NaturalElementTextFormat{71}{174.97}{Lu}{Lutetium}};

%% Actinide
  \node[name=Ac, below of=La, LanthanideActinide, yshift=-1cm] {\NaturalElementTextFormat{89}{227}{Ac}{Actinium}};
  \node[name=Th, right of=Ac, LanthanideActinide] {\NaturalElementTextFormat{90}{232.04}{Th}{Thorium}};
  \node[name=Pa, right of=Th, LanthanideActinide] {\NaturalElementTextFormat{91}{231.04}{Pa}{Protactinium}};
  \node[name=U, right of=Pa, LanthanideActinide] {\NaturalElementTextFormat{92}{238.03}{U}{Uranium}};
  \node[name=Np, right of=U, LanthanideActinide] {\NaturalElementTextFormat{93}{237}{Np}{Neptunium}};
  \node[name=Pu, right of=Np, LanthanideActinide] {\NaturalElementTextFormat{94}{244}{Pu}{Plutonium}};
  \node[name=Am, right of=Pu, LanthanideActinide] {\NaturalElementTextFormat{95}{243}{Am}{Americium}};
  \node[name=Cm, right of=Am, LanthanideActinide] {\NaturalElementTextFormat{96}{247}{Cm}{Curium}};
  \node[name=Bk, right of=Cm, LanthanideActinide] {\NaturalElementTextFormat{97}{247}{Bk}{Berkelium}};
  \node[name=Cf, right of=Bk, LanthanideActinide] {\NaturalElementTextFormat{98}{251}{Cf}{Californium}};
  \node[name=Es, right of=Cf, LanthanideActinide] {\NaturalElementTextFormat{99}{252}{Es}{Einsteinium}};
  \node[name=Fm, right of=Es, LanthanideActinide] {\NaturalElementTextFormat{100}{257}{Fm}{Fermium}};
  \node[name=Md, right of=Fm, LanthanideActinide] {\NaturalElementTextFormat{101}{258}{Md}{Mendelevium}};
  \node[name=No, right of=Md, LanthanideActinide] {\NaturalElementTextFormat{102}{259}{No}{Nobelium}};
%  \node[name=Lr, right of=No, LanthanideActinide] {\NaturalElementTextFormat{103}{262}{Lr}{Lawrencium}};

%% Draw dotted lines connecting Lanthanide breakout to main table
  \draw (LaLu.north west) edge[dotted] (La.north west)
        (LaLu.north east) edge[dotted] (Yb.north east)
        (LaLu.south west) edge[dotted] (La.south west)
        (LaLu.south east) edge[dotted] (Yb.south east);
%% Draw dotted lines connecting Actinide breakout to main table
  \draw (AcLr.north west) edge[dotted] (Ac.north west)
        (AcLr.north east) edge[dotted] (No.north east)
        (AcLr.south west) edge[dotted] (Ac.south west)
        (AcLr.south east) edge[dotted] (No.south east);

%% Legend
  \draw[black, AlkaliMetalFill] ($(La.north -| Fr.west) + (35em,45.0em)$)
    rectangle +(2em, 2em) node[right, yshift=-2.5ex, xshift=2ex]{\Huge {\boldmath \bf $s$ subshells: $\ell = 0$}};

  \draw[black, MetalFill] ($(La.north -| Fr.west) + (35em,41.0em)$)
    rectangle +(2em, 2em) node[right, yshift=-2.5ex, xshift=2ex]{\Huge {\boldmath \bf $d$ subshells: $\ell = 1$}};

  \draw[black, NonmetalFill] ($(La.north -| Fr.west) + (35em,37.0em)$)
    rectangle +(2em, 2em) node[right, yshift=-2.5ex, xshift=2ex]{\Huge {\boldmath \bf $p$ subshells: $\ell = 2$}};

  \draw[black, LanthanideActinideFill] ($(La.north -| Fr.west) + (35em,33.0em)$)
    rectangle +(2em, 2em) node[right, yshift=-2.5ex, xshift=2ex]{\Huge {\boldmath \bf $f$ subshells: $\ell = 3$}};

\iffalse
  \node at ($(La.north -| Fr.west) + (5em,-15em)$) [name=elementLegend, Element, fill=white]
    {\NaturalElementTextFormat{Z}{mass}{Symbol}{Name}};
  \node[Element, fill=white, right of=elementLegend, xshift=1em]
    {\NaturalElementTextFormat{}{}{man-made}{}} ;

%% Diagram Title
  \node at (H.west -| Fe.north) [name=diagramTitle, TitleLabel]
    {(Mendeleev's) Periodic Table of Chemical Elements via Ti\emph{k}Z};
\fi

\end{tikzpicture}
\vskip 1em
\caption{Periodic table illustrating the Madelung rules.}
\label{figure_2}
\end{figure}

\newpage

\vbox{\vskip 3em}

\begin{table}[h!]
\renewcommand{\arraystretch}{1.2}
\begin{tabular}{|c|c|c|c|c|l|}
\hline
{\bf atomic numbers} & {\bf subshell} & {\boldmath $n + \ell$} & {\boldmath $n$} & {\boldmath $\ell$} & {\boldmath $E_{n,\ell}$}  \\
\hline
\color{darkorange} \bf 1--2  & \color{darkorange} \bf  1s & \color{darkorange}\bf 1 & \color{darkorange} \bf 1 & \color{darkorange} \bf 0 & \color{darkorange} \boldmath $4$ 
\\  
\hline
\color{darkorange} \bf 3--4  & \color{darkorange} \bf 2s & \color{darkorange} \bf 2 & \color{darkorange} \bf 2 & \color{darkorange} \bf 0 & \color{darkorange} \boldmath $6$ 
\\
\hline
\color{darkgreen} \bf 5--10   & \color{darkgreen} \bf 2p & \color{darkgreen} \bf 3 & \color{darkgreen} \bf 2 & \color{darkgreen} \bf 1 & \color{darkgreen} \boldmath $7\frac{1}{3}$ 
\\
\color{darkorange} \bf 11--12 & \color{darkorange} \bf 3s  & \color{darkorange} \bf 3 & \color{darkorange} \bf 3 & \color{darkorange} \bf 0 & \color{darkorange} \boldmath $8$  
\\
\hline
\color{darkgreen} \bf 13--18 & \color{darkgreen} \bf 3p & \color{darkgreen} \bf 4 & \color{darkgreen} \bf 3 & \color{darkgreen} \bf 1 & \color{darkgreen} \boldmath $9\frac{1}{3}$  
\\
\color{darkorange} \bf 19--20 & \color{darkorange} \bf 4s & \color{darkorange} \bf 4 & \color{darkorange} \bf 4 & \color{darkorange} \bf 0 & \color{darkorange} \boldmath $10$ \\
\hline
\color{darkblue} \bf 21--30 & \color{darkblue} \bf 3d & \color{darkblue} \bf 5 & \color{darkblue} \bf 3 & \color{darkblue} \bf2 & \color{darkblue} \boldmath $11\frac{1}{5}$ \\
\color{darkgreen} \bf  31--36 & \color{darkgreen} \bf  4p & \color{darkgreen} \bf  5 & \color{darkgreen} \bf  4 & \color{darkgreen} \bf 1 & \color{darkgreen} \boldmath $11\frac{1}{3}$ \\
\color{darkorange} \bf 37--38 & \color{darkorange} \bf 5s & \color{darkorange} \bf 5 &  \color{darkorange} \bf 5 & \color{darkorange} \bf 0 & \color{darkorange} \boldmath $12$ \\
\hline
\color{darkblue} \bf 39--48  & \color{darkblue} \bf 4d & \color{darkblue} \bf 6 & \color{darkblue} \bf 4 & \color{darkblue} \bf 2 & \color{darkblue} \boldmath $13\frac{1}{5}$ \\
\color{darkgreen} \bf  49--54  & \color{darkgreen} \bf 5p & \color{darkgreen} \bf 6 & \color{darkgreen} \bf  5 & \color{darkgreen} \bf  1 & \color{darkgreen} \boldmath $13\frac{1}{3}$ \\
\color{darkorange} \bf 55--56  & \color{darkorange} \bf 6s & \color{darkorange} \bf 6 & \color{darkorange} \bf 6 & \color{darkorange} \bf 0 & \color{darkorange} \boldmath $14$ \\
\hline
\color{darkpink} \bf 57--70 & \color{darkpink} \bf 4f  & \color{darkpink} \bf7 & \color{darkpink} \bf 4 & \color{darkpink} \bf 3 & \color{darkpink} \boldmath $15\frac{1}{7}$ \\
\color{darkblue} \bf 71--80 & \color{darkblue} \bf 5d & \color{darkblue} \bf 7 & \color{darkblue} \bf 5 & \color{darkblue} \bf 2 & \color{darkblue} \boldmath $15\frac{1}{5}$ \\
\color{darkgreen} \bf  81--86 & \color{darkgreen} \bf  6p & \color{darkgreen} \bf 7 & \color{darkgreen} \bf  6 & \color{darkgreen} \bf  1 & \color{darkgreen} \boldmath $15\frac{1}{3}$ \\
\color{darkorange} \bf 87--88 & \color{darkorange} \bf 7s & \color{darkorange} \bf 7 & \color{darkorange} \bf 7 & \color{darkorange} \bf 0 & \color{darkorange} \boldmath $16$ \\
\hline
\color{darkpink} \bf 89--102  & \color{darkpink} \bf 5f & \color{darkpink} \bf 8 & \color{darkpink} \bf 5 & \color{darkpink} \bf 3 & \color{darkpink} \boldmath $17\frac{1}{7}$ \\
\color{darkblue} \bf 103--112 & \color{darkblue} \bf 6d & \color{darkblue} \bf 8 & \color{darkblue} \bf 6 & \color{darkblue} \bf 2 & \color{darkblue} \boldmath $17\frac{1}{5}$ \\
\color{darkgreen} \bf 113--118 & \color{darkgreen} \bf  7p & \color{darkgreen} \bf  8 & \color{darkgreen} \bf 7 & \color{darkgreen} \bf 1 & \color{darkgreen} \boldmath $17\frac{1}{3}$ \\
\color{darkorange} \bf 119--120 & \color{darkorange} \bf 8s & \color{darkorange} \bf 8 & \color{darkorange} \bf 8 & \color{darkorange} \bf 0 & \color{darkorange} \boldmath $18$ \\
\hline
\end{tabular}

\vskip 1em \noindent
\caption{The Madelung rules for elements 1 to 120, and the energies $E_{n,\ell}$.}
\label{table_1}
\end{table}
\vskip 1em

\iffalse
\break
{\color{darkorange} \bf \boldmath $s$ subshells, with $\ell = 0$, are in orange.}
\break
  {\color{darkgreen} \bf \boldmath $p$ subshells, with $\ell = 1$, are in green.}
\break
  {\color{darkblue} \bf \boldmath $d$ subshells, with $\ell = 2$, are in blue.}
\break
 {\color{darkpink} \bf \boldmath $f$ subshells, with $\ell = 3$, are in pink.}
\fi

\newpage

\appendix

\section{Proofs}
\label{sec:unitary_equivalence}

Here we give the proofs of Theorems \ref{thm:unitary_equivalence_1} and \ref{thm:unitary_equivalence_2}.

\vskip 1em

\noindent \textbf{Theorem \ref{thm:unitary_equivalence_1}.} \textit{The operator $F \maps L^2(S^3) \otimes \C^2 \to \H$ is a unitary equivalence between the representation $\rho$ of $\SU(2) \times \SU(2)$ on $L^2(S^3) \otimes \C^2$ and the representation $\rho_\H$ of this group on $\H$, and $F \d = \d F$ on the domain of $\d$.}

\begin{proof}  
First we prove that $\d$ commutes with $F$ on the domain of $\d$.  For this, note that 
\[   F =  I(p_+ + C p_-) \]
where 
\[   p_+, p_- \maps  L^2(S^3) \otimes \C^2 \to L^2(S^3) \otimes \C^2 \]
are the self-adjoint projections to the $+1$ and $-1$ eigenspaces of $S = \d/|\d|$, and $I$ is the identity operator regarded as a map from $L^2(S^3) \otimes \C^2$ to $\H$, which is the same space with a different complex structure.  Clearly $\d$ commutes with $I$ and with $p_+$ and $p_-$.   To show that $\d$  commutes with $F$, it thus suffices to show that $\d$ commutes with $C$.

For this, recall that 
\[      
C \maps L^2(S^3) \otimes \C^2 \to   L^2(S^3) \otimes \C^2  
\]
is defined by
\[       (C \psi)(g) = \epsilon \, \overline{\psi}(g)   \]
where $\overline{\psi}$ is the componentwise complex conjugate of $\psi \maps S^3 \to \C^2$ and  
\[      \epsilon = \left(\begin{array}{cc} 0 & 1 \\ -1 & 0 \end{array} \right) .\]
Since $\d = i X_j \otimes \sigma_j + \tfrac{3}{2}$ by Equation \eqref{eq:dirac_operator}, to show $\d$ commutes with $C$ it  is enough to show that $iX_j \otimes \sigma_j$ commutes with $C$.   For brevity we write $X_j$ for $X_j \otimes 1$ and $\sigma_j$ for $1 \otimes \sigma_j$; these operators commute so we can write $iX_j \otimes \sigma_j$ as $i\sigma_j X_j$.     Thus, to show $\d$ commutes with $C$ it suffices to show
\begin{equation}
\label{eq:to_be_shown}
  \epsilon \, \overline{i \sigma_j X_j \psi} = i \sigma_j X_j \epsilon \overline{\psi}  
\end{equation}
for all square-integrable $\psi \maps S^3 \to \C^2$, where the derivative is in the distributional sense.
 
It is easy to check that
\[     - \epsilon \, \overline{\sigma}_j =  \sigma_j \, \epsilon \]
for $j = 1, 2, 3$.   This implies Equation \eqref{eq:C_commutation} and it also implies Equation \eqref{eq:to_be_shown}, as desired:
\[  \begin{array}{ccc}  
\epsilon \, \overline{i \sigma_j X_j \psi} 
&=& i \sigma_j \,\epsilon \, \overline{X_j \psi} \\ [2pt]
&=& i \sigma_j \, \epsilon X_j \overline{\psi} \\ [2pt]
&=& i \sigma_j  X_j \, \epsilon \,\overline{\psi} .
\end{array}
\]

Next we prove that $F$ is complex linear.  For this we show that $j F =  F i$
as operators from $L^2(S^3) \otimes \C^2$ to $\H$, or, using the definition of $F$,
\[    j I(p_+ + C p_-) =  I(i  p_+ + C i p_-) . \]
%We have seen that $C$ commutes with $\d$.   Thus $C$ preserves the $+1$ and $-1$ eigenspaces of $S$, so 
%\[           C p_{\pm} = p_{\pm} C.\]
Since $j$ equals $i$ on the $+1$ eigenspace of $S$ and $-i$ on the $-1$ eigenspace, we have
\[            j I p_{\pm} = \pm I i p_{\pm} .\]
Using this fact and the conjugate-linearity of $C$, we obtain
\[
\begin{array}{ccl}
 j I(p_+ + C p_-)
 &=&  I( i p_+ - i C p_-)  \\
&=& I(i  p_+ + C i p_-) 
\end{array} 
\]
as desired.   

Next we prove that $F$ is unitary.    Since $C$ commutes with $\d$ and thus $S = \d/|\d|$, it follows that $C$ and therefore $F$ preserve the $+1$ and $-1$ eigenspaces of $S$.  To show $F$ is unitary, it thus suffices to show it is unitary as an operator on each of these subspaces.   Since $F$ is the identity on the $+1$ eigenspace of $S$, we just need to show it is unitary when restricted to the $-1$ eigenspace of $S$, where it equals $I C$.   For this, first note that $I C$ preserves the norm on the $-1$ eigenspace of $S$: given a vector $\psi_-$ in this subspace we have
\[     \| I C \psi_- \|^2 
= \int_{S^3} \left\langle \epsilon \overline{\psi_-}(g) , \epsilon \overline{\psi_-}(g) \right\rangle \, dg 
=  \int_{S^3} \left\langle \psi_-(g) , \psi_-(g) \right\rangle \, dg 
= \| \psi_- \|^2 .\]
Second, note that $I C$ is invertible: indeed, $C$ is its own inverse and $I$ is invertible since it is the identity on the underlying sets of the vector spaces it goes between.   

Finally we prove that $F$ is an equivalence of $\SU(2) \times \SU(2)$ representations.  For this it suffices to show
\[        F \rho(g_1, g_2) \psi = \rho_\H(g_1,g_2) F \psi \]
for all $\psi \in  L^2(S^3) \otimes \C^2$.  Indeed, it suffices to check this equation for $\psi_+ = p_+ \psi$ and $\psi_- = p_- \psi$ separately.   First, recall from Section \ref{sec:einstein} that the operators $\rho(g_1,g_2)$ commute with $\d$, so they preserve the eigenspaces of $S = \d/|\d|$, and thus
\[           \rho(g_1, g_2) \psi_\pm = (\rho(g_1,g_2) \psi)_\pm. \]
Second, note that the definition of $\rho_\H$ says
\[             I \rho(g_1, g_2) = \rho_\H(g_1,g_2) I .\]
For $\psi_+$ we thus have
\[  \begin{array}{ccl}
      F \rho(g_1,g_2)(\psi_+) &=& F ((\rho(g_1,g_2) \psi)_+) \\ [2pt]
      &=& I((\rho(g_1,g_2) \psi)_+) \\ [2pt]
     &=& \rho_\H(g_1,g_2) (I(\psi_+))  \\ [2pt]
     &=& \rho_\H(g_1,g_2) F (\psi_+) .
\end{array}
\]
For $\psi_-$ the calculation is a bit longer:
\[  \begin{array}{ccl}
      F \rho(g_1,g_2)(\psi_-) &=& F ((\rho(g_1,g_2) \psi)_-) \\ [2pt]
      &=& I C((\rho(g_1,g_2) \psi)_-) \\ [2pt]
      &=& I C((\rho(g_1,g_2) (\psi_-)) \\ [2pt]
      &=& I (\rho(g_1,g_2) (C(\psi_-)) \\ [2pt]
      &=& (\rho(g_1,g_2) (I C (\psi_-)) \\ [2pt]
      &=& \rho_\H(g_1,g_2) (F(\psi_-)).
\end{array}
\]
Here the fourth equation uses the fact that $C$ commutes with the operators $\rho(g_1,g_2)$.  
To see this, note that for any $\psi \in L^2(S^3) \otimes \C^2$, 
\[  \begin{array}{ccl}
 C(\rho(g_1,g_2) \psi)(g) &=& \epsilon \, \overline{g_2 \psi(g_1^{-1} g g_2)}
 \\  [3pt]
&=& g_2 \epsilon \, \overline{\psi(g_1^{-1} g g_2)}  \\ [3pt]
&=& (\rho(g_1,g_2) (C\psi))(g)
\end{array}
\] 
where in the second step we use Equation \eqref{eq:C_commutation}. 
\end{proof}

\vskip 1em

\noindent{\textbf{Theorem \ref{thm:unitary_equivalence_2}.} \textit{The map $\Fock(F) \maps \Fock \h \to \Fock \H$ is a unitary equivalence between the representation $\Fock(\rho)$ of the group $\SU(2) \times \SU(2)$ on the fermionic Fock space  $\Fock \h$ for bound states of the hydrogen atom Hamiltonian and the representation $\Fock(\rho_\H)$ of this group on the fermionic Fock space $\Fock \H$ for left-handed massless spin-$\hf$ particles on the Einstein universe.   That is,
\[      \Fock(F) \Fock(\rho(g_1,g_2)) = \Fock(\rho_\H(g_1,g_2)) \Fock(F)  \]
for all $(g_1,g_2) \in \SU(2) \times \SU(2)$.   Moreover 
\[   \Fock(F) d\Fock(\d) = d\Fock(\d) \Fock(F)  \]
on the domain of $d\Fock(\d)$.}

\begin{proof}
If $V \maps \H_1 \to \H_2$ and $U \maps \H_2 \to \H_3$ are any unitary operators between Hilbert spaces, we have
\[   \begin{array}{ccl}
\Fock(UV) (\phi_1 \wedge \cdots \wedge \phi_n) &=& UV\phi_1 \wedge \cdots \wedge UV\phi_n \\ [3pt]
&=& \Fock(U)(V \phi_1 \wedge \cdots \wedge V \phi_n) \\ [3pt]
&=& \Fock(U) \Fock(V) ( \phi_1 \wedge \cdots \wedge \phi_n)  
\end{array}
\]
for all $\phi_1, \dots \phi_n \in \H_1$.  Since linear combinations of wedge products $\phi_1 \wedge \cdots \wedge \phi_n$ are dense in $\Fock(\H_1)$, and all the operators involved are continuous, we
conclude that $\Fock(UV) = \Fock(U) \Fock(V)$.  Similarly $\Fock(1) = 1$, and thus $\Fock(U)^{-1} = \Fock(U^{-1}) $.   

Using this together with Theorem \ref{thm:unitary_equivalence_1}, it follows that
\[  \begin{array}{ccl}
 \Fock(F) \Fock(\rho(g_1,g_2)) \Fock(F)^{-1} &=& \Fock( F \rho(g_1,g_2) F^{-1}) \\ [3pt]
 &=& \Fock(\rho_\H(g_1,g)) 
 \end{array}
\]
and thus 
\[      \Fock(F) \Fock(\rho(g_1,g_2)) = \Fock(\rho_\H(g_1,g_2)) \Fock(F)  \]
for all $g \in \SU(2) \times \SU(2)$, as was to be shown.

By Theorem \ref{thm:unitary_equivalence_1} we also have $\d = F \d F^{-1}$, where we must beware that $\d$ on the left is a self-adjoint operator on $\h$ while $\d$ on the right is a self-adjoint operator on $\H$.   Since $F$ is unitary, this implies that for any $t \in \R$ we have
\[     \exp(it\d) = F \exp(i t \d) F^{-1}  \]
and thus
\[    
\begin{array}{ccl}
 \Fock(\exp(i t \d)) &=& \Fock(F \exp(i t \d) F^{-1}) \\ [3pt]
 &=&\Fock(F) \Fock(\exp(i t \d)) \Fock(F)^{-1}  \\ [3pt]
 &=& \Fock(F) \exp(i t d\Fock(\d)) \Fock(F)^{-1}
 \end{array}
 \]
where in the last step we use Equation \eqref{eq:fock_exponentiation}.   
Differentiating with respect to $t$ and setting $t = 0$, we get
\[    d\Fock(\d) =  \Fock(F) d\Fock(\d) \Fock(F)^{-1}     \]
on the domain of $d\Fock(\d)$, and thus
\[    \Fock(F) d \Fock(\d) = d\Fock(\d) \Fock(F) \]
as was to be shown.
\end{proof}

\end{document}